%% file: main.tex
\documentclass[11pt]{article}
\usepackage{etex}
\usepackage[latin1]{inputenc}
\usepackage{amsmath, amsthm, amssymb}
\usepackage[english]{babel}
\usepackage{pdfpages}
\usepackage{times}

\usepackage{thmtools}
\usepackage{thm-restate}

\usepackage{colortbl}
\usepackage{graphicx}
\usepackage{caption}
\usepackage{subcaption}
\usepackage{longtable}
\usepackage{booktabs}
\usepackage{tabu}
\usepackage[referable]{threeparttablex}
\usepackage{multirow}
\usepackage{hyperref}
\usepackage[capitalise]{cleveref}

\crefname{figure}{Figure}{Figure}
\crefformat{footnote}{#2\footnotemark[#1]#3}
\usepackage{tikz}
\usepackage{comment}
\usepackage{parskip}

\usepackage{todonotes}
\usepackage{algorithmic}
\usepackage[ruled, linesnumbered]{algorithm2e}
\usepackage{tikz}
\usetikzlibrary{arrows,automata,shapes,decorations,decorations.markings,calc, matrix,decorations.pathmorphing}

\usepackage{enumerate,paralist}

\newtheorem{theorem}{Theorem}
\newtheorem{corollary}{Corollary}
\newtheorem{lemma}{Lemma}

\newtheorem{proposition}{Proposition}

\theoremstyle{remark}
\newtheorem{remark}{Remark}

\theoremstyle{definition}

\usepackage{fullpage}

\usepackage{graphicx} 
\newcommand{\ov}{\overline}
\newcommand{\Nats}{\mathbb{N}}
\newcommand{\NatsPlus}{\Nats^+}

\newcommand{\Realsplus}{\mathbb{R}_{\geq 0}}

\newcommand{\restr}[1]{[#1]}

\newcommand{\Path}{\rightsquigarrow}

\newcommand{\Prb}{\mathbb{P}}
\newcommand{\Prob}[1]{\mathbb{P}[#1]}
\newcommand{\Probr}[1]{\mathbb{P}[#1]}

\newcommand{\ProbT}[1]{\mathbb{P}^{\Temp}[#1]}

\newcommand{\Expect}[1]{\mathsf{E}[#1]}

\newcommand{\eps}{\varepsilon}

\newcommand{\MC}{\mathcal{M}}
\newcommand{\MCStates}{\mathcal{X}}
\newcommand{\trans}{\delta}
\newcommand{\Walk}{\rho}

\newcommand{\Fitness}{\mathsf{F}}
\newcommand{\Degree}{\mathsf{deg}}
\newcommand{\Weight}{\mathsf{w}}
\newcommand{\Neighbors}{\mathsf{Nh}}
\newcommand{\Temp}{\mathsf{T}}
\newcommand{\Uniform}{\mathsf{U}}
\newcommand{\State}{\mathsf{X}}
\newcommand{\ModState}{\ov{\State}}
\newcommand{\FixProb}{\mathsf{\rho}}
\newcommand{\FixProbU}{\mathsf{f}^\Uniform}

\newcommand{\FixEv}{\mathcal{E}}
\newcommand{\Size}{\mathsf{size}}
\newcommand{\ev}{\mathcal{E}}

\newcommand{\SelfLoop}{\mathsf{sl}}
\newcommand{\Graphseq}{\mathcal{G}}
\newcommand{\Diameter}{\mathrm{diam}}
\newcommand{\SpanTree}{\mathcal{T}}

\newcommand{\Sink}{\mathcal{H}}

\newcommand{\Frontier}{\mathcal{F}}
\newcommand{\Distance}{\lambda}
\newcommand{\Parent}{\mathsf{par}}
\newcommand{\Children}{\mathsf{chl}}
\newcommand{\Ancestors}{\mathsf{anc}}
\newcommand{\Descendants}{\mathsf{des}}
\newcommand{\HeterState}{\mathcal{H}}
\newcommand{\DeadState}{\mathcal{D}}
\newcommand{\coeff}{\eta}
\newcommand{\In}{\mathsf{In}}
\newcommand{\Out}{\mathsf{Out}}

\newcommand\numberthis{\addtocounter{equation}{1}\tag{\theequation}}

\newif\iffullproofs
\newif\ifshortproofs
\fullproofstrue 
\shortproofsfalse

\title{Strong Amplifiers of Natural Selection: Proofs}

\author
{
Andreas Pavlogiannis$^{1}$, Josef Tkadlec$^{1}$, 
Krishnendu Chatterjee$^{1}$, Martin A. Nowak$^{2}$\\
\\
\normalsize{$^{1}$IST Austria, A-3400 Klosterneuburg, Austria}\\
\normalsize{$^{2}$Program for Evolutionary Dynamics, Department of  Organismic and Evolutionary Biology,}\\
\normalsize{Department of Mathematics, Harvard University, Cambridge, MA 02138, USA}\\
\\
}

\date{}

\begin{document}
\maketitle

\input{abstract}

\tableofcontents

\input{introduction}

\input{model}
\input{preliminaries}

\input{bounded_temp}
\input{amplifiers}

{
\bibliographystyle{abbrv}
\bibliography{bibliography}
}

\end{document}

%% file: abstract.tex
\begin{abstract}
We consider the modified Moran process on graphs to study the spread of genetic and cultural mutations on structured populations. An initial mutant arises either spontaneously (aka \emph{uniform initialization}), or during reproduction (aka \emph{temperature initialization}) in a population of $n$ individuals, and has a fixed fitness advantage $r>1$ over the residents of the population. The fixation probability is the probability that the mutant takes over the entire population. Graphs that ensure fixation probability of~1 in the limit of infinite populations are called  \emph{strong amplifiers}. Previously, only a few examples of strong amplifiers were known for uniform initialization, whereas no strong amplifiers were known for temperature initialization.

In this work, we study necessary and sufficient conditions for strong amplification, and prove negative and positive results. We show that for temperature initialization, graphs that are unweighted and/or self-loop-free have fixation probability upper-bounded by $1-1/f(r)$, where $f(r)$ is a function linear in $r$. Similarly, we show that for uniform initialization, bounded-degree graphs that are unweighted and/or self-loop-free have fixation probability upper-bounded by $1-1/g(r,c)$, where $c$ is the degree bound and $g(r,c)$ a function linear in $r$. Our main positive result complements these negative results, and is as follows:  every family of undirected graphs with (i)~self loops and (ii)~diameter bounded by $n^{1-\epsilon}$, for some fixed $\epsilon>0$, can be assigned weights that makes it a strong amplifier, both for uniform and temperature initialization.
\end{abstract}

%% file: introduction.tex
\section{Introduction}\label{sec:intro}

\noindent{\em The Moran process.}
Evolutionary dynamics study the change of population over time under the effect of natural selection and random drift~\cite{Nowak06b}.
The Moran process~\cite{Moran1962} is an elegant stochastic model for the rigorous study of how mutations spread in a population.
Initially, a population of $n$ individuals, called the residents, exists in a homogeneous state, and a random individual becomes mutant.
The mutants are associated with a fitness advantage $r\geq 1$, whereas the residents have fitness normalized to~1.
The Moran process is a discrete-time stochastic process, described as follows.
In every step, a single individual is chosen for reproduction with probability proportional to its fitness.
This individual produces a single offspring (a copy of itself), which replaces another individual chosen uniformly at random from the population.
The main quantity of interest is the \emph{fixation probability} $\FixProb(n, r)$, defined as the probability that the single invading mutant will eventually take over the population. 
As typically $r$ is small (i.e., $r=1+\epsilon$, for some small $\epsilon>0$) and $n$ is large, we study the fixation probability at the limit of large populations,
i.e., $\FixProb(r)=\lim_{n\to\infty}\FixProb(n,r)$.
It is known that $\FixProb(r)=1-r^{-1}$.

\noindent{\em The Moran process on graphs.}
The standard Moran process takes place on \emph{well-mixed} populations where the reproducing individual can replace any other in the population.
However, natural populations have spatial structure, where each individual has a specific set of neighbors, and mutation spread must respect this structure.
Evolutionary graph theory represents spatial structure as a (generally weighted, directed) graph, where each individual occupies a vertex of the graph, and edges define interactions between neighbors~\cite{Lieberman05}.
The Moran process on graphs is similar to the standard Moran process, with the exception that the offspring replaces a neighbor of the reproducing individual.
The well-mixed population is represented by the complete graph $K_n$.
If the graph is strongly connected, the Moran process is guaranteed to reach a homogeneous state where mutants either fixate or go extinct.

\noindent{\em Mutant initialization.}
The asymmetry introduced by the population structure makes the fixation probability depend on the placement of the initial mutant. 
In \emph{uniform initialization}, the initial mutant arises \emph{spontaneously}, i.e., uniformly at random on each vertex.
In \emph{temperature initialization}, the initial mutant arises\emph{during reproduction} i.e., on each vertex with probability proportional to the rate that the vertex is replaced by offspring from its neighbors.
Hence our interested is on the fixation probability $\FixProb(G_n^{\Weight}, r, Z)$ for a weighted graph $G_n^{\Weight}$ of $n$ vertices and under initialization $Z\in\{\Uniform, \Temp\}$, denoting uniform and temperature initialization, respectively.

\noindent{\em Amplifiers of selection.}
Population structure affects the fixation probability of mutants. 
An infinite family of graphs $(G_n^{\Weight})_n$ is \emph{amplifying} for initialization $Z$ if $\lim_{n\to\infty}\FixProb(G_n^{\Weight}, r, Z)>1-r^{-1}$,
%i.e., at the limit of infinite population size, the fixation probability is larger than that of the fixation probability of a mutant with the same fitness on a well mixed population.
Intuitively, the fitness advantage of mutants is being ``amplified'' by the structure compared to the well-mixed population.
\emph{Strong amplifying} families have $\lim_{n\to\infty}\FixProb(G_n^{\Weight}, r, Z)=1$, and hence ensure the fixation of mutants.
On the other hand, \emph{bounded amplifiers} have $\lim_{n\to\infty}\FixProb(G_n^{\Weight}, r, Z)\leq 1-1/f(r)$, where $f$ is a linear function,
and hence provide limited amplification at best.

\noindent{\em Existing results.}
The Moran process on graphs was introduced in~\cite{Lieberman05}, where several amplifying and strongly amplifying families were presented.
Under uniform initialization, the canonical example is the family of undirected Star graphs, with fixation probability $1-r^{-2}$, making it a \emph{quadratic uniform amplifier}~\cite{Lieberman05,Broom2008,Monk2014}.
Among directed graphs, strongly amplifying families are known to exist: 
(i)~Superstars and Metafunnels were already introduced in~\cite{Lieberman05}, where their strong amplifying properties were outlined, and
(ii)~more recently, the family of Megastars was rigorously proved to be a strong amplifying family~\cite{Galanis17}.
Megastars were subsequently shown to be optimal (up to logarithmic factors) wrt the rate that fixation probability converges to~1 as a function of $n$~\cite{Goldberg17}.
Among undirected graphs, the family of Stars was the best amplifying family know for a long time, and the existence of strong amplifiers was open.
Recently, undirected strong amplifiers were presented independently in~\cite{Goldberg17} and~\cite{Giakkoupis16}.
 
Under temperature initialization, the landscape is more scarce.
None of the uniform amplifiers mentioned in the previous paragraph is a temperature amplifier.
It turns out that on all those structures the mutants go extinct with high probability when the initial placement is according to temperature.
Recently, the Looping Star family was introduced in~\cite{Adlam15} and was shown to be a quadratic amplifier under both initialization schemes.
Crucially, Looping Stars contain self-loops and weights.
To our knowledge, no other temperature amplifier has been known.

\noindent{\em Our contributions.}
In this work, we study necessary and sufficient conditions for strong amplifiers, and prove negative and positive results.
\begin{compactenum}
\item Our negative results are as follows.
For temperature initialization, we show that graphs which are unweighted and/or self-loop-free have fixation probability upper-bounded by $1-1/f(r)$, where $f(r)$ is a function linear in $r$. Hence, without both weights and self-loops, there are only bounded temperature amplifiers.
Similarly, we show that for uniform initialization, bounded-degree graphs that are unweighted and/or self-loop-free have fixation probability upper-bounded by $1-1/g(r,c)$, where $c$ is the degree bound and $g(r,c)$ a function linear in $r$.
Hence, without both weights and self-loops, bounded-degree graph families are only bounded uniform amplifiers.
\item Our positive result complements these negative results and is as follows.
We show that every family of undirected graphs with (i)~self loops and (ii)~diameter bounded by $n^{1-\epsilon}$, for some fixed $\epsilon>0$, 
can be assigned weights that makes the family a strong amplifier, both for uniform and temperature initialization.
Moreover, the weight construction requires $O(n)$ time.
\end{compactenum}

Our proof techniques rely on the analysis of Markov chains, the Cauchy-Schwarz inequality, concentration bounds, stochastic domination and coupling arguments.
The weight construction in our positive result is straightforward, however proving the amplification properties of the resulting structure is more involved.

\subsection{Other Related Work}\label{subsec:other_related}
Strong amplifiers were already introduced in~\cite{Lieberman05}, however it was later shown that the fixation probability on Superstars is weaker than originally stated, and hence the heuristic argument for strong amplification cannot be made formal~\cite{Diaz13}.
In~\cite{Galanis17}, it was shown that the fixation probability on Superstars as appeared in~\cite{Lieberman05} is indeed too optimistic, by proving an upper bound on the rate that the probability can tend to~1 as a function of $n$.
A revised analysis of Superstars appeared in~\cite{Hauert}.
The work of~\cite{Pavlogiannis17} introduced the Metastars as a family of unweighted undirected graphs with better amplification properties than Stars, for specific values of the fitness advantage $r$.
Other aspects of the Moran process on graphs have also been studied in the literature.
In~\cite{Mertzios13}, the authors studied undirected suppressors of selection, which are graphs that suppress the selective advantage of mutants, as opposed to amplifying it.
Recently, a family of strong suppressors was presented~\cite{Giakkoupis16}.
The work of~\cite{Mertzios13b} studies selective amplifiers, a notion that characterizes the number of initial vertices that guarantee mutant fixation.
Randomly structured populations were shown to have no effect on fixation probability in~\cite{Adlam2014}.
Besides the fixation probability, the absorption time of the Moran process is crucial for characterizing the rate of evolution~\cite{Frean2013b}
and has been studied on various graphs~\cite{Diaz16}.
Finally, computational aspects of computing the fixation probability on graphs were studied in~\cite{CSPPL}, where the problem was shown to admit a
fully polynomial randomized approximation scheme, later improved in~\cite{Chatterjee17}.

%% file: model.tex
\section{Organization}
The organization of this document is as follows: Before presenting our proofs we present the detailed description of our model and the results in Section 2.
We then present the formal notation (Section 3), the proofs of our negative results (Section 4) and the proofs of our positive results (Section 5).

\section{Model and Summary of Results}

\subsection{Model}

\smallskip\noindent{\em The birth-death Moran process.}
The \textit{Moran process} considers a population of $n$ individuals, which undergoes
reproduction and death, and each individual is either a resident or a mutant~\cite{Moran1962}.
The residents and the mutants have constant fitness 1 and $r$, respectively. 
The Moran process is a discrete-time stochastic process defined as follows: in the initial step, a single mutant is introduced 
into a homogeneous resident population. 
At each step, an individual is chosen randomly for reproduction 
with probability proportional to its fitness; another individual is chosen uniformly at random for death and is replaced by 
a new individual of the same type as the reproducing individual. 
Eventually, this Markovian process ends when all individuals become of one of the two types.
The probability of the event that all individuals become mutants is called the {\em fixation} probability. 

\smallskip\noindent{\em The Moran process on graphs.}
In general, the Moran process takes place on a population structure, which is represented
as a graph.
The vertices of the graph represent individuals and edges represent interactions between individuals \cite{Lieberman05, Nowak06b}. 
Formally, let $G_n=(V_n,E_n,W_n)$ be a weighted, directed graph, 
where $V_n=\{1,2,\ldots,n\}$ is the vertex set , $E_n$ is the Boolean edge matrix, and $W_n$ is a stochastic weight matrix.
An edge is a pair of vertices $(i,j)$ which is indicated by $E_n[i,j]=1$ and denotes that there is an interaction from $i$ to $j$
(whereas we have $E_n[i,j]=0$ if there is no interaction from $i$ to $j$).
The stochastic weight matrix $W_n$ assigns weights to interactions, i.e., $W_n[i,j]$ is positive iff $E_n[i,j]=1$,
and for all $i$ we have $\sum_{j} W_n[i,j]=1$.
For a vertex $i$, we denote by $\In(i)=\{ j \mid E_n[j,i]=1\}$ (resp., $\Out(i)=\{ j \mid E_n[i,j]=1\}$) 
the set of vertices that have incoming (resp., outgoing) interaction or edge to (resp., from) $i$.
Similarly to the Moran process, at each step an individual is chosen randomly for reproduction with probability proportional to its fitness.
An edge originating from the reproducing vertex is selected randomly with probability equal to its weight. 
The terminal vertex of the chosen edge takes on the type of the vertex at the origin of the edge.
In other words, the stochastic matrix $W_n$ is the weight matrix that represents the choice probability 
of the edges.
We only consider graphs which are {\em connected}, i.e., every pair of vertices is connected by a path.
This is a sufficient condition to ensure that in the long run, the Moran process reaches a homogeneous state
(i.e., the population consists entirely of individuals of a single type).
\iffullproofs
See Figure~\ref{fig:moran} for an illustration.
\fi
The well-mixed population is represented by a complete graph where all edges 
have equal weight of $1/n$.

\iffullproofs
\begin{figure}
\centering
\setlength{\unitlength}{\textwidth} 
\begin{picture}(1,0.875)
\put(0.05,0){\includegraphics[width=0.9\unitlength]{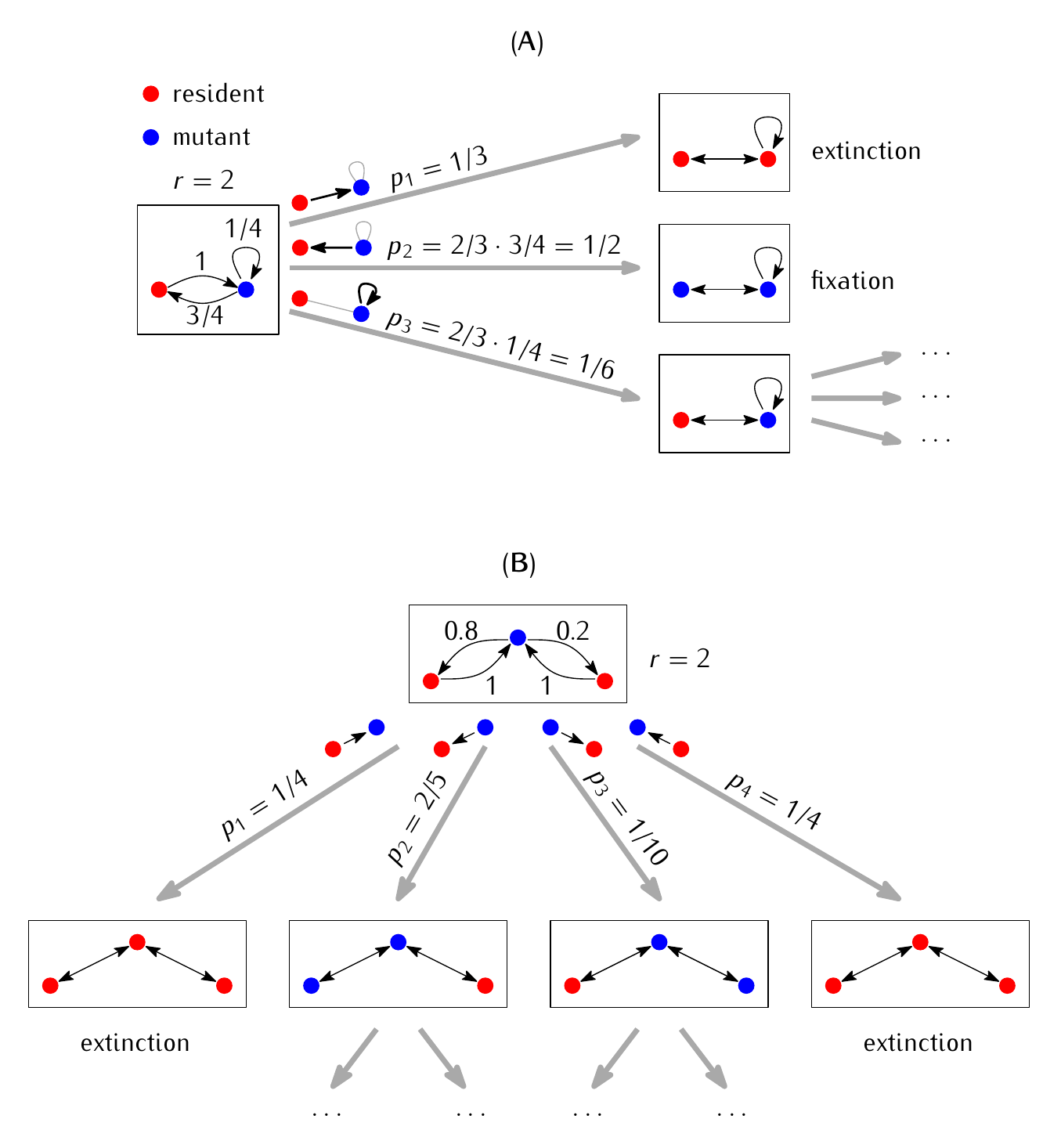}}
\end{picture}
\caption{
Illustration of one step of the Moran process on a 
weighted graph with self-loops. Residents are depicted as red vertices, and 
mutants as blue vertices. 
As a concrete example, we consider the relative fitness of the mutants 
is $r=2$.
In Figure~\ref{fig:moran}(A), the total fitness of the population is $\mathcal{F}=1+2=3$, and hence the probability of selecting resident (resp., mutant) for reproduction equals $1/3$ (resp., $2/3$). 
The mutant reproduces along an edge, and the edge is chosen randomly 
proportional to the edge weight.
Figure~\ref{fig:moran}(B) shows that different reproduction events might 
lead to the same outcome.
}
\label{fig:moran}
\end{figure}
\fi

\smallskip\noindent{\em Classification of graphs.}
We consider the following classification of graphs:
\begin{enumerate}
\item {\em Directed vs undirected graphs.}
A graph $G_n=(V_n,E_n,W_n)$ is called {\em undirected} if for all $1\leq i,j \leq n$
we have $E_n[i,j]=E_n[j,i]$. 
In other words, there is an edge from $i$ to $j$ iff there is an edge from $j$ to $i$,
which represents symmetric interaction.
If a graph is not undirected, then it is called a {\em directed} graph.

\item {\em Self-loop free graphs.}
A graph $G_n=(V_n,E_n,W_n)$ is called a {\em self-loop free} graph iff for all 
$1 \leq i \leq n$ we have $E_n[i,i]=W_n[i,i]=0$.

\item {\em Weighted vs unweighted graphs.}
A graph $G_n=(V_n,E_n,W_n)$ is called an {\em unweighted} graph if for all 
$1 \leq i \leq n$ we have 
\[
W_n[i,j]=
\begin{cases}
\frac{1}{|\Out(i)|}  & j \in \Out(i);\\
0 & j \not\in \Out(i)
\end{cases}
\]
In other words, in unweighted graphs for every vertex the edges are choosen 
uniformly at random. Note that for unweighted graphs the weight matrix is not
relevant, and can be specified simply by the graph structure $(V_n,E_n)$.
In the sequel, we will represent unweighted graphs as $G_n=(V_n,E_n)$.
%If a graph is not unweighted, we call it a weighted graph.

\item {\em Bounded degree graphs.}
The degree of a graph $G_n=(V_n,E_n,W_n)$, denoted $\deg(G_n)$, is 
$\max\{ \In(i),\Out(i) \mid 1 \leq i \leq n\}$, i.e., the maximum in-degree 
or out-degree. 
For a family of graphs $(G_n)_{n > 0}$ we say that the family has bounded degree, 
if there exists a constant $c$ such that the degree of all graphs in the family
is at most $c$, i.e., for all $n$ we have $\deg(G_n) \leq c$. 
\end{enumerate}

\smallskip\noindent{\em Initialization of the mutant.}
The fixation probability is affected by many different factors \cite{Patwa2008}. 
In a well-mixed population, the fixation probability depends on the population size $n$ and 
the relative fitness advantage $r$ of mutants~\cite{Maruyama1974a, Nowak06b}. 
For the Moran process on graphs, the fixation probability also depends on 
the population structure, which breaks the symmetry and homogeneity 
of the well-mixed population
% \cite{levin1974disturbance, levin1976population, durrett1994stochastic, Lieberman05, Broom2008, Frean2013, Barton1993, Whitlock2003, Houchmandzadeh2011}. 
\cite{levin1974disturbance, levin1976population, durrett1994stochastic, Lieberman05, Broom2008, Frean2013, Whitlock2003, Houchmandzadeh2011}. 
Finally, for general population structures, the fixation probability typically depends on the initial 
location of the mutant \cite{Allen2014, Antal2006}, unlike the well-mixed population 
where the probability of the mutant fixing is independent of where the mutant 
arises \cite{Maruyama1974a, Nowak06b}. 
There are two standard ways mutants may arise in a population~\cite{Lieberman05,Adlam15}.
First, mutants may arise spontaneously and with equal probability at any vertex of the population structure. 
In this case we consider that the mutant arise at any vertex uniformly at random and 
we call this \textit{uniform initialization}. 
Second, mutants may be introduced through reproduction, and thus arise at a vertex with rate proportional to the incoming
edge weights of the vertex.
We call this \textit{temperature initialization}. 
In general, uniform and temperature initialization result in different fixation probabilities.

\smallskip\noindent{\em Amplifiers, quadratic amplifiers, and strong amplifiers.}
Depending on the initialization, a population structure can distort fitness differences \cite{Lieberman05, Nowak06b, Broom2008}, 
where the well-mixed population serves as a canonical point of comparison. 
Intuitively, amplifiers of selection exaggerate variations in fitness by increasing 
(respectively decreasing) the chance of fitter (respectively weaker) mutants fixing 
compared to their chance of fixing in the well-mixed population. 
In a well-mixed population of size $n$, the fixation probability is 
\[
\frac{1-1/r}{1-(1/r)^n}.
\]
Thus, in the limit of large population (i.e., as $n \to \infty$) 
the fixation probability in a well-mixed population is $1-1/r$.
We focus on two particular classes of amplifiers that are of special interest.
A family of graphs $(G_n)_{n>0}$ is a {\em quadratic} amplifier if in the limit
of large population the fixation probability is $1-1/r^2$.
Thus, a mutant with a 10\% fitness advantage over the resident has approximately the 
same chance of fixing in quadratic amplifiers as a mutant with a 21\% fitness advantage 
in the well-mixed population.
A family of graphs $(G_n)_{n> 0}$ is an {\em arbitrarily strong} amplifier (hereinafter called simply a strong amplifier) if
for any constant $r> 1$ the fixation probability approaches~1 at the limit of large population sizes,
whereas when $r<1$, the fixation probability approaches~0.
There is a much finer classification of amplifiers presented in~\cite{Adlam15}. 
 %For other amplifiers (such as cubic, polynomial amplifiers) see~\cite{Adlam15}.
We focus on quadratic amplifiers which are the most well-known among polynomial
amplifiers, and strong amplifiers which represent the strongest
form of amplification.

Amplifiers tend to have fixation times longer than the well mixed population. Therefore they are  especially useful in situations where the rate limiting step is the discovery and evaluation of marginally advantageous mutants. An interesting direction for future work would be to consider amplifiers as well as the time-scale of evolutionary trajectories.

\smallskip\noindent{\em Existing results.} 
We summarize the main existing results in terms of uniform and temperature initialization.
\begin{enumerate}
\item {\em Uniform initialization.} 
First, consider the family of Star graphs, which consist of one central vertex and $n-1$ leaf vertices,
with each leaf being connected to and from the central vertex.
Star graphs are unweighted, undirected, self-loop free graphs, whose degree is linear in the population size.
Under uniform initialization, the family of Star graphs is a 
quadratic amplifier~\cite{Lieberman05, Nowak06b}.
A generalization of Star graphs, called Superstars~\cite{Lieberman05,Nowak06b,Hauert,CSPPL}, are known to be 
strong amplifiers under uniform initialization~\cite{Galanis17}. 
The Superstar family consists of unweighted, self-loop free, 
but directed graphs where the degree is linear in the population size.
Another family of directed graphs with strong amplification properties, called Megastars, was recently introduced in~\cite{Galanis17}.
The Megastars are stronger amplifiers than the Superstars, as the fixation probability on the former is a approximately $1-n^{-1/2}$ (ignoring logarithmic factors),
and is asymptotically optimal (again, ignoring logarithmic factors).
In contrast, the fixation probability on the Superstars is approximately $1-n^{-1/2}$.
In the limit of $n\to \infty$, both families approach the fixation probability~1.

\item {\em Temperature initialization.} 
While the family of Star graphs is a quadratic amplifier under uniform 
initialization, it is not even an amplifier under temperature 
initialization~\cite{Adlam15}. 
It was shown in~\cite{Adlam15} that by adding self-loops and weights to the edges 
of the Star graph, a graph family, namely the family of Looping Stars, can be constructed, 
which is a quadratic amplifier simultaneously under temperature and uniform initialization.
Note that in contrast to Star graphs, the Looping Star graphs are weighted and also have self-loops.

\end{enumerate}

\smallskip\noindent{\em Open questions.} 
Despite several important existing results on amplifiers of selection, several basic questions have remained open:
\begin{enumerate}
\item {\em Question~1.} Does there exist a family of self-loop free graphs
(weighted or unweighted) that is a quadratic amplifier under temperature 
initialization? 

\item {\em Question~2.} Does there exist a family of unweighted graphs 
(with or without self-loops) that is a quadratic amplifier under temperature
initialization?

\item {\em Question~3.} Does there exist a family of bounded degree 
self-loop free (weighted or unweighted) graphs that is a strong amplifier under 
uniform initialization? 

\item {\em Question~4.} Does there exist a family of bounded degree unweighted graphs 
(with or without self-loops) that is a strong amplifier under uniform
initialization?

\item {\em Question~5.} Does there exist a family of graphs that is  
a strong amplifier under temperature initialization?
More generally, does there exist a family of graphs that is  
a strong amplifier both under temperature and uniform initialization?

\end{enumerate}
To summarize, the open questions ask for (i)~the existence of quadratic amplifiers
under temperature initialization without the use of self-loops, or weights 
(Questions~1 and~2); (ii)~the existence of strong amplifiers
under uniform initialization without the use of self-loops, or weights, and while the degree of the graph is small;
 and (iii)~the existence of strong amplifiers under temperature initialization.
While the answers to Question~1 and Question~2 are positive under uniform initialization, 
they have remained open under temperature initialization.
Questions~3 and~4 are similar to~1 and~2, but focus on uniform rather than temperature initialization.
The restriction on graphs of bounded degree is natural: large degree means that some individuals 
must have a lot of interactions, whereas graphs of bounded degree represent simple structures.
Question~5 was mentioned as an open problem in~\cite{Adlam15}.
Note that under temperature initialization, even the existence of a cubic 
amplifier, that achieves fixation probability at least $1-(1/r^3)$ in the 
limit of large population, has been open~\cite{Adlam15}.

\subsection{Results}
In this work we present several negative as well as positive results that answer
the open questions (Questions~1-5) mentioned above. 
We first present our negative results.

\smallskip\noindent{\em Negative results.} 
Our main negative results are as follows:
\begin{enumerate}
\item 
Our first result (Theorem~1) shows that for any self-loop free weighted graph 
$G_n=(V_n,E_n,W_n)$, for any $r\geq 1$, under temperature initialization the 
fixation probability is at most $1-1/(r+1)$. 
The implication of the above result is that it answers Question~1 in negative.

\item 
Our second result (Theorem~2) shows that for any unweighted (with or 
without self-loops) graph $G_n=(V_n,E_n)$, for any $r \geq 1$, under temperature 
initialization the fixation probability is at most $1-1/(4r+2)$. 
The implication of the above result is that it answers Question~2 in negative.

\item 
Our third result (Theorem~3) shows that for any bounded degree 
self-loop free graph (possibly weighted) $G_n=(V_n,E_n,W_n)$, for any $r \geq 1$, 
under uniform initialization the fixation probability is at most $1-1/(c+c^2r)$,
where $c$ is the bound on the degree, i.e., $\deg(G_n) \leq c$. 
The implication of the above result is that it answers Question~3 in negative.

\item 
Our fourth result (Theorem~4) shows  that for any unweighted, bounded degree 
graph (with or without self-loops) $G_n=(V_n,E_n)$, for any $r \geq 1$, 
under uniform initialization the fixation probability is at most $1-1/(1+r c)$,
where $c$ is the bound on the degree, i.e., $\deg(G_n) \leq c$. 
The implication of the above result is that it answers Question~4 in negative.

\end{enumerate}

\smallskip\noindent{\em Significance of the negative results.}
We now discuss the significance of the above results. 
\begin{enumerate}
\item The first two negative results show that in order to obtain quadratic amplifiers 
under temperature initialization, self-loops and weights are inevitable, complementing 
the existing results of~\cite{Adlam15}. 
More importantly, it shows a sharp contrast between temperature and uniform 
initialization: while self-loop free, unweighted graphs (namely, Star graphs) are quadratic
amplifiers under uniform initialization, no such graph families are quadratic
amplifiers under temperature initialization.

\item The third and fourth results show that without using self-loops and weights, bounded degree graphs 
cannot be made strong amplifiers even under uniform initialization.
See also \cref{rem:bounded_degree_uniform}.
\end{enumerate}

\smallskip\noindent{\em Positive result.}
Our main positive result shows the following:
\begin{enumerate}
\item For any constant $\epsilon>0$, consider any connected unweighted graph $G_n=(V_n,E_n)$ 
of $n$ vertices with self-loops and which has {\em diameter} at most $n^{1-\epsilon}$.
The diameter of a connected graph is the maximum, among all pairs of vertices, of 
the length of the shortest path between that pair.
We establish (Theorem~5) that there is a stochastic weight matrix $W_n$ 
%assigning positive weights only to $E_n$ 
such that for any $r>1$ the fixation probability on $G_n=(V_n,E_n,W_n)$
both under uniform and temperature initialization is at least $1-\frac{1}{n^{\epsilon/3}}$.
An immediate consequence of our result is the following: 
for any family of connected unweighted graphs with self-loops $(G_n=(V_n,E_n))_{n> 0}$ 
such that the diameter of $G_n$ is at most 
$n^{1-\epsilon}$, for a constant $\epsilon>0$, one can construct a stochastic weight 
matrix $W_n$ such that the resulting family $(G_n=(V_n,E_n,W_n))_{n > 0}$ of 
weighted graphs is a strong amplifier simultaneously under uniform 
and temperature initialization.
Thus we answer Question~5 in affirmative.
\end{enumerate}
%In the Methods section we describe the basic principle (see Figure~\ref{fig:dynamics}) 
%and illustrate on several simple graphs (see Figure~\ref{fig:graphs}).

\smallskip\noindent{\em Significance of the positive result.}
We highlight some important aspects of the results established in this work.
\begin{enumerate}
\item First, note that for the fixation probability of the Moran process on graphs to be well defined,
a necessary and sufficient condition is that the graph is connected.
A uniformly chosen random connected unweighted graph of $n$ vertices has diameter bounded by a constant,
with high probability.
Hence, within the family of connected, unweighted graphs,
the family of graphs of diameter at most $O(n^{1-\epsilon})$, for any constant $0<\epsilon<1$,
has probability measure~1.
Our results establish a strong dichotomy: (a)~the negative results
state that without self-loops and/or without weights, {\em no} family of 
graphs can be a quadratic amplifier (even more so a strong amplifier) 
even for only temperature initialization; and 
(b)~in contrast, for {\em almost all} families of connected graphs with 
self-loops, there exist weight functions such that the resulting family 
of weighted graphs is a strong amplifier both under temperature 
and uniform initialization.
 
\item Second, with the use of self-loops and weights, even simple 
graph structures, such as Star graphs, Grids, and well-mixed structures (i.e., complete graphs) can be made strong amplifiers.

\item Third, our positive result is constructive, rather than existential.
In other words, we not only show the existence of 
strong amplifiers, but present a construction of them.

%\item Finally, note that in using weights, edges can be effectively removed
%by assigning to their weight a small value.
%However, edges cannot be created. 
%Thus, for complete graphs, desired sub-graphs can be created easily using weights.
%Our positive result states that for almost all graphs, one can use weights to create sub-graphs
%which are strong amplifiers both under uniform and temperature initialization.

\end{enumerate} 

Our results are summarized in \cref{table}.

\begin{remark}
\smallskip\noindent{\em Edges with zero weight.}
Note that edges can be effectively removed by being assigning zero weight (however, no weight assignment can create edges that don't exist.) Therefore, when our construction works for some graph, it also works for a graph that contains some additional edges. In particular, our construction easily works for complete graphs. The construction can also be extended to a scenario in which we insist that each edge is assigned a positive (non-zero) weight.
\end{remark}

\begin{table}
\begin{center}
\begin{tabular}{|c|c|c|c|c|}
\hline
  & \multicolumn{2}{|c|}{Temperature}  & \multicolumn{2}{|c|}{Uniform$^\star$} \\ 
\cline{2-5} 
 & Loops & No Loops  & Loops & No Loops \\
\hline
Weights & $\checkmark$ & $\times$ & $\checkmark$& $\times$ \\
\hline
No Weights & $\times$ & $\times$ & $\times$ & $\times$ \\
\hline
\end{tabular}
\end{center}
\caption{
Summary of  our results on existence of strong 
amplifiers for different initialization schemes (temperature initialization 
or uniform initialization) and graph families (presence or absence of loops 
and/or weights). The ``$\checkmark$'' symbol marks that for given choice 
of initialization scheme and graph family, almost all graphs admit a weight 
function that makes them strong amplifiers. 
The ``$\times$'' symbol marks that for given choice of initialization scheme 
and graph family, no strong amplifiers exist (under any weight function). 
The asterisk signifies that the negative results under uniform initialization 
only hold for bounded degree graphs.
}
\label{table}
\end{table}

%% file: preliminaries.tex
\section{Preliminaries: Formal Notation}\label{sec:preliminaries}

\subsection{The Moran Process on Weighted Structured Populations}\label{subsec:moran_process}
We consider a population of $n$ individuals on a graph $G_n=(V_n,E_n,W_n)$.
Each individual of the population is either a \emph{resident}, or a \emph{mutant}.
Mutants are associated with a \emph{reproductive rate} (or \emph{fitness}) $r$, whereas the reproductive rate of residents is normalized to $1$.
Typically we consider the case where $r>1$, i.e., mutants are \emph{advantageous}, whereas when $r<1$ we call the mutants \emph{disadvantageous}.
We now introduce the formal notation related to the process.

\smallskip\noindent{\em Configuration.}
A \emph{configuration} of $G_n$ is a subset $S \subseteq V$ which specifies the vertices of $G_n$ 
that are occupied by mutants and thus the remaining vertices $V\setminus S$ are occupied by residents.
We denote by $\Fitness(S)=r\cdot |S| + n-|S|$  the 
total fitness of the population in configuration $S$, where $|S|$ is the number of mutants in $S$.

\smallskip\noindent{\em The Moran process.}
The birth-detah Moran process on $G_n$ is a discrete-time Markovian random process.
We denote by $\State_i$ the random variable for a configuration at time step $i$,
and $\Fitness(\State_i)$ and $|\State_i|$ denote the total fitness and the number of
mutants of the corresponding configuration, respectively.
The probability distribution for the next configuration $\State_{i+1}$ at time $i+1$
is determined by the following two events in succession: 
\begin{compactdesc}
\item[{\em Birth:}] One individual is chosen at random to reproduce, with probability proportional to its fitness.
That is, the probability to reproduce is $r/\Fitness(\State_i)$ for a mutant, and $1/\Fitness(\State_i)$ for a resident. 
%%where $\Fitness(S_i)=r\cdot |S_i| + n-|S_i|$ is the total fitness of the population in configuration $S_i$.
Let $u$ be the vertex occupied by the reproducing individual.
\item[{\em Death:}] A neighboring vertex $v\in \Out(u)$ is chosen randomly with probability $W_n[u,v]$. 
The individual occupying $v$ dies, and the reproducing individual places a copy of its own on $v$.
Hence, if $u\in \State_i$, then $\State_{i+1}=\State_i\cup \{v\}$, otherwise $\State_{i+1}=\State_i\setminus\{v\}$.
\end{compactdesc}
The above process is known as the \emph{birth-death} Moran process, where the death event is conditioned on the birth event,
and the dying individual is a neighbor of the reproducing one.
%For $\State_i$, we denote by $\Fitness(\State_i)$ the total fitness of the population in $\State_i$
%and by $|\State_i|$ the number of mutants in $\State_i$.
%%When $r>1$, mutants reproduce more rapidly than residents.
%We are interested in the behavior of the Moran process on graphs when mutants have a small selective advantage, 
%i.e. $r=1+\eps$ for some $\eps\geq 0$, and at the limit of large populations, i.e., as $n\to\infty$.

\smallskip\noindent{\em Probability measure.}
Given a graph $G_n$ and the fitness $r$, the birth-death Moran process defines 
a probability measure on sequences of configurations, which we denote as 
$\Prb^{G_n,r}[\cdot]$. 
If the initial configuration is $\{u\}$, then we define the probability measure
as $\Prb^{G_n,r}_u[\cdot]$, and if the graph and fitness $r$ is clear from the 
context, then we drop the superscript. 

\smallskip\noindent{\em Fixation event.}
The fixation event, denoted $\FixEv$, represents that all vertices are mutants,
i.e., $\State_i=V$ for some $i$.
In particular, $\Prb^{G_n,r}_u[\FixEv]$ denotes the fixation probability in $G_n$
for fitness $r$ of the mutant, when the initial mutant is placed on vertex $u$.
We will denote this fixation probability as $\FixProb(G_n,r,u)=\Prb^{G_n,r}_u[\FixEv]$.
 
\subsection{Initialization and Fixation Probabilities}
We will consider three types of initialization, namely, (a)~uniform initialization, where
the mutant arises at vertices with uniform probability, (b)~temperature initialization,
where the mutant arises at vertices proportional to the temperature, and (c)~convex combination 
of the above two.

\smallskip\noindent{\em Temperature.} 
For a weighted graph $G_n=(V_n,E_n,W_n)$, the temperature of a vertex $u$, denoted 
$\Temp(u)$, is $\sum_{v \in \In(u)} W_n[v,u]$, i.e., the sum of the incoming weights.
Note that $\sum_{u \in V_n} \Temp(u)=n$, and a graph is {\em isothermal} iff $\Temp(u)=1$ 
for all vertices $u$.

\smallskip\noindent{\em Fixation probabilities.}
We now define the fixation probabilities under different initialization.

\begin{compactenum}
\item {\em Uniform initialization.} The fixation probability under uniform initialization is
\[
\FixProb(G_n,r,\Uniform) =\sum_{u \in V_n} \frac{1}{n} \cdot \FixProb(G_n,r,u).
\]

\item {\em Temperature initialization.} The fixation probability under temperature initialization is
\[
\FixProb(G_n,r,\Temp) =\sum_{u \in V_n} \frac{\Temp(u)}{n} \cdot \FixProb(G_n,r,u). 
\]

\item {\em Convex initialization.} In $\coeff$-\emph{convex initialization}, where $\coeff\in[0,1]$, 
the initial mutant arises with probability $(1-\coeff)$ via uniform initialization, and with probability $\coeff$ 
via temperature initialization.
The fixation probability is then 
\[
\FixProb(G_n,r,\coeff) =(1-\coeff) \cdot \FixProb(G_n,r,\Uniform) + \coeff \cdot \FixProb(G_n,r,\Temp).
\]
\end{compactenum}

%Intuitively, the temperature $\Temp(u)$ of a node $u$ captures the rate at which individuals from the neighborhood $\Neighbors(u)$ place offsprings on $u$.
%In temperature initialization, the initial mutant is thought of being introduced in the population during a reproduction event,
%where a mutant offspring is placed on a node $u$ via a reproducing resident individual occupying a neighboring node $v\in \Neighbors(u)$.
%Hence, nodes with higher temperature are more likely to attract the initial mutant.
%In contrast, uniform initialization models situations in which every individual is equally likely to turn to a mutant. 
%Finally, the $\coeff$-convex initialization captures situations where the initial mutant arises with probability $\coeff$ uniformly, and with probability $(1-\coeff)$ through reproduction.

\subsection{Strong Amplifier Graph Families}

A \emph{family} of graphs $\Graphseq$ is an infinite sequence of weighted graphs $\Graphseq=(G_n)_{n\in \NatsPlus}$.
%We call $\Graphseq$ \emph{self-loop free} if each $G_n^{\Weight_n}$ is self-loop free, and call $\Graphseq$ \emph{unweighted} if each $G_n^{\Weight_n}$ is unweighted.
%Additionally, we call $\Graphseq$ a family of \emph{bounded degree} if there exists a constant $c$ such that for every graph $G_n=(V_n,E_n)$ in the family, we have that $\max_{u\in V_n} |\Neighbors(u)|\leq c$, i.e., the degree of $G_n$ is bounded by $c$.
\begin{compactitem}
\item {\em Strong amplifiers.} 
A family of graphs $\Graphseq$ is a \emph{strong uniform amplifier} (resp. \emph{strong temperature amplifier}, \emph{strong convex amplifier}) 
if for every fixed $r_1>1$ and $r_2<1$ we have that
\[
\liminf_{n\to\infty} \FixProb(G_n,r_1,Z) = 1 \qquad \text{and} \qquad \limsup_{n\to\infty} \FixProb(G_n,r_2,Z) = 0\ ;
\]
where $Z=\Uniform$ (resp., $Z=\Temp$, $Z=\coeff$).
\end{compactitem}
Intuitively, strong amplifiers ensures (a)~fixation of advantageous mutants with probability~1 and 
(b)~extinction of disadvantageous mutants with probability~1.
In other words, strong amplifiers represent the strongest form of amplifiers possible.

%% file: bounded_temp.tex
\section{Negative Results}\label{sec:bounded_temp}
In the current section we present our negative results, which show the nonexistence of strong amplifiers in 
the absence of either self-loops or weights.
In our proofs, we consider weighted graph $G_n=(V_n,E_n,W_n)$, and for notational simplicity we drop
the subscripts from vertices, edges and weights, i.e., we write $G_n=(V,E,W)$.
We also consider that $G_n$ is connected and $n\geq 2$.
Throughout this section we will use a technical lemma, which we present below.
Given a configuration $\State_i=\{u\}$ with one mutant,
let $x$ and $y$ be the probability that in the next configuration the mutants increase and go extinct, respectively.
The following lemma bounds the fixation probability $\FixProb(G_n,r,u)$ as a function of $x$ and $y$.
%\ifshortproofs
%We refer to \cite{FullReport} for the proof.
%\fi

\begin{lemma}\label{lemm:technical}
Consider a vertex $u$ and the initial configuration $\State_0=\{u\}$ where the
initial mutant arises at vertex $u$.
For any configuration $\State_i=\{u\}$, let 
\[
x=\Prb^{G_n,r}[|\State_{i+1}|=2 \mid \State_i=\{u\}]
\qquad \text{and} \quad
y=\Prb^{G_n,r}[|\State_{i+1}|=0 \mid \State_i=\{u\}]\ .
\]
be the probability that the number of mutants increases (or decreases) in a single step.
Then the fixation probability from $u$ is at most $x/(x+y)$,
i.e., 
\[
\FixProb(G_n,r,u) \leq \frac{x}{x+y} = 1 -\frac{y}{x+y}\ .
\]
\end{lemma}
\ifshortproofs
\begin{proof}
Let's focus on the first event that changes the number of mutants. The probability that this event decreases the number of mutants equals $\frac{y}{x+y}$. In such case, the mutants have gone extinct, hence the extinction probability is at least $\frac{y}{x+y}$ and the fixation probability is at most $\frac{x}{x+y}$. \end{proof}
\fi
\iffullproofs
\begin{proof}
We upperbound the fixation probability $\FixProb(G_n,r,u)$ starting from $u$ by the probability that a configuration 
$\State_t$ is reached with $|\State_t|=2$. 
Note that to reach fixation the Moran process must first reach a configuration with at least two mutants.
We now analyze the probability to reach at least two mutants.
This is represented by a three-state one dimensional random walk, where two states are absorbing, 
one absorbing state represents a configuration with two mutants, and the other absorbing state represents 
the extinction of the mutants, and the bias towards the absorbing state representing two mutants is $x/y$.
See \cref{fig:three_state_rw} for an illustration.
Using the formulas for absorption probability in one-dimensional three-state Markov chains (see, e.g.,~\cite{Kemeny12},~\cite[Section~6.3]{Nowak06b}), 
we have the probability that a configuration with two mutants is reached is 
\[
\frac{1 -(x/y)^{-1}}{1-(x/y)^{-2}}= \frac{1}{1 + (x/y)^{-1}}= \frac{x}{x+y}\ .
\]
Hence it follows that $\FixProb(G_n,r,u) \leq 1-\frac{y}{x+y}$.
\end{proof}
\input{three_state_rw}
\fi

\subsection{Negative Result~1}\label{subsec:bounded_loopless_temperatue}
We now prove our negative result~1.

\smallskip
\begin{theorem}\label{them:loopless_bounded_temp}
For all self-loop free graphs $G_n$ and for every $r\geq 1$ we have  
$\FixProb(G_n,r,\Temp)\leq 1-1/(r+1)$.
\end{theorem}
\ifshortproofs
\begin{proof}
Since $G_n$ is self-loop free, for all $u$ we have $W[u,u]=0$. 
Hence $ \Temp(u) = \sum_{v\in \In(u)\setminus\{u\}} W[v,u]$.
Consider the case where the initial mutant is placed on vertex $u$.
Using Lemma~\ref{lemm:technical},
%Then the fixation probability from $u$ is at most the probability that the mutant spreads 
%to two vertices before it goes extinct. 
a simple calculation shows that the fixation probability $\FixProb(G_n,r,u)$ 
from $u$ is at most
%%this event happens with probability at most
\[
1-\frac{\Temp(u)}{\Temp(u)+r}\ .
\]
Summing over all vertices $u$, we obtain
\[
\FixProb(G_n,r,\Temp) = \sum_{u} \frac{\Temp(u)}{n} \cdot \FixProb(G_n,r,u) \leq 
\frac{1}{n}\cdot \sum_{u} \Temp(u)\cdot \left(1-\frac{\Temp(u)}{\Temp(u)+r}\right)  \leq 1-\frac{1}{r+1}\ .
\]
where the first inequality uses the bound above, and the last inequality is obtained using the Cauchy-Schwarz inequality
in the form
$$\sum_u \frac{\Temp(u)^2}{\Temp(u)+r} \geq \frac{\left(\sum_u \Temp(u)\right)^2}{\sum_u (\Temp(u)+r)}=\frac{n}{r+1}.
$$
%A detailed calculation is available in~\cite{FullReport}.
The desired result follows.
\end{proof}
\fi
\iffullproofs
\begin{proof}
Since $G_n$ is self-loop free, for all $u$ we have $W[u,u]=0$. 
Hence $ \Temp(u) = \sum_{v\in \In(u)\setminus\{u\}} W[v,u]$.
Consider the case where the initial mutant is placed on vertex $u$, i.e, 
$\State_0=\{u\}$.
For any configuration $\State_i=\{u\}$, we have the following:
\[
x=\Prb^{G_n,r}[|\State_{i+1}|=2 \mid \State_i=\{u\}] = \frac{r}{\Fitness(\State_i)} 
\] 
%\qquad  \text{ and } \qquad 
\[
y=\Prb^{G_n,r}[|\State_{i+1}|=0 \mid \State_i=\{u\}] =
\frac{1}{\Fitness(\State_i)}\cdot \sum_{v\in \In(u)\setminus\{u\}} W[v,u] = 
\frac{1}{\Fitness(\State_i)}\cdot \Temp(u)\ .
\]
Thus $x/y= r/\Temp(u)$.
%\[
%\frac{\Probr{|\State_{i+1}|=2|\State_i=\{u\}}}{\Probr{|\State_{i+1}|=0|\State_i=\{u\}}} = \frac{r}{\Temp(u)}
%\numberthis\label{eq:ratio_u}
%\]
%We upperbound the fixation probability $\FixProb(G_n,r,u)$ starting from $u$ by the probability that a configuration $\State_t$ is reached with $|\State_t|=2$. 
%This is reperesented by a three-state one dimensional random walk, where two states are absorbing, one absorbing 
%state represents a configuration with two mutants, other absorbing state represents the extinction of the
%mutants, and the bias towards the absorbing state representing two mutants is $x/y=r/\Temp(u)$.
%Using the formulas for absoption probability in one-dimensional three-state Markov chains, 
%we have the probability that a configuration with two mutants is reached is 
%\[
%\frac{1-\left(\frac{r}{\Temp(u)}\right)^{-1}}{1-\left(\frac{r}{\Temp(u)}\right)^{-2}}.
%\]
Hence by Lemma~\ref{lemm:technical} we have
\[
\FixProb(G_n,r,u) \leq 1-\frac{\Temp(u)}{\Temp(u)+r}\ .
\]
%\[ \FixProb(G_n,r,u) \leq \frac{1-\left(\frac{r}{\Temp(u)}\right)^{-1}}{1-\left(\frac{r}{\Temp(u)}\right)^{-2}} = \frac{1}{1+\frac{\Temp(u)}{r}} = 1-\frac{\Temp(u)}{\Temp(u)+r}\]
Summing over all $u$, we obtain
\[
\FixProb(G_n,r,\Temp)=\sum_{u} \frac{\Temp(u)}{n} \cdot \FixProb(G_n,r,u) \leq \frac{1}{n}\cdot \sum_{u} \Temp(u)\cdot \left(1-\frac{\Temp(u)}{\Temp(u)+r}\right) = 1-\frac{1}{n}\cdot \sum_u \frac{\Temp(u)^2}{\Temp(u)+r}\ ;
\numberthis\label{eq:temp_upperbound}
\]
since $\sum_u\Temp(u)=n$.
Using the Cauchy-Schwarz inequality, we obtain
\[
\sum_u \frac{\Temp(u)^2}{\Temp(u)+r}\geq \frac{\left(\sum_u\Temp(u)\right)^2}{\sum_u (\Temp(u) + r)} = \frac{n^2}{n+n\cdot r} = \frac{n}{r+1}\ ;
\]
and thus \cref{eq:temp_upperbound} becomes
\[
\FixProb(G_n,r,\Temp) \leq 1-\frac{1}{n}\cdot \frac{n}{r+1} = 1-\frac{1}{r+1}\; 
\]
as desired.
\end{proof}
\fi

We thus arrive at the following corollary.
\smallskip
\begin{corollary}\label{cor:loopless_bounded_temp}
There exists no self-loop free family of graphs which is a strong temperature amplifier.
\end{corollary}

\subsection{Negative Result~2}\label{subsec:bounded_weightless_temperature}
We now prove our negative result~2. 
%%of the main text.
\smallskip
\begin{theorem}\label{them:unweighted_bounded_temp}
For all unweighted graphs $G_n$ and for every $r\geq 1$ we have
$\FixProb(G_n,r,\Temp) \leq 1-1/(4r+2)$.
\end{theorem}
\ifshortproofs
\begin{proof}
For every vertex $u\in V$, let
\[
\Temp'(u) = \sum_{v\in \In(u)\setminus\{u\}} \frac{1}{|\Out(v)|}\ .
\]
We have the following two inequalities for $\Temp'$  
\[
\Temp(u) \geq \Temp'(u) \qquad \text{ and } \qquad
\sum_u\Temp'(u) \geq  \frac{n}{2}\ .
\]
Consider the case where the initial mutant is placed on vertex  $u$.
Using Lemma~\ref{lemm:technical},
%Then the fixation probability from $u$ is at most the probability that the mutant spreads to two vertices before it goes extinct. 
a simple calculation shows that the fixation probability $\FixProb(G_n,r,u)$  from $u$ is at most
%%is event happens with probability at most
\[
1-\frac{\Temp'(u)}{\Temp'(u)+r}.
\]
Summing over all vertices $u$, we obtain
\[
\FixProb(G_n,r,\Temp) =  \sum_{u} \frac{\Temp(u)}{n}\cdot \FixProb(G_n,r,u) \leq 
\frac{1}{n}\cdot \sum_{u} \Temp(u)\cdot \left(1-\frac{\Temp'(u)}{\Temp'(u)+r}\right)  \leq 1-\frac{1}{4r+2}\ ;
\]
where the last inequality follows from the inequalities for $\Temp'$, the Cauchy-Schwarz inquality in the form
%we obtain the inequalities above using the inequalities for $\Temp'$, the bound on fixation 
%probability, and the Cauchy-Schwarz inequality
$$\sum_u \frac{\Temp'(u)^2}{\Temp'(u)+r} \geq \frac{\left(\sum_u \Temp'(u)\right)^2}{\sum_u (\Temp'(u)+r)}$$
and from the fact that the function $\frac{x^2}{x+n\cdot r}$ is increasing in $x$ for $x>0$ and any $r,n>0$.
%A detailed calculation is available in~\cite{FullReport}.
The desired result follows.
\end{proof}
\fi
\iffullproofs
\begin{proof}
For every vertex $u\in V$, let
\[
\Temp'(u) = \sum_{v\in \In(u)\setminus\{u\}} \frac{1}{|\Out(v)|}\ .
\]
We establish two inequalities related to $\Temp'$.
Since $G_n$ is unweighted, we have
\[
\Temp(u) = \sum_{v\in \In(u)} \frac{1}{|\Out(v)|}\geq \Temp'(u)\ .
\]
For a vertex $u$, let $\SelfLoop(u)=1$ if $u$ has a self-loop and $\SelfLoop(u)=0$ otherwise.
Since $G_n$ is connected, each vertex $u$ has at least one neighbor other than itself.
Thus for every vertex $u$ with $\SelfLoop(u)=1$ we have that $|\Out(u)|\geq 2$.
Hence 
\begin{align*}
\sum_u\Temp'(u) = &   \sum_u\left (\sum_{v\in \In(u)} \frac{1}{|\Out(v)|} -\SelfLoop(u)\frac{1}{|\Out(u)|}\right)
=\sum_u\left (\sum_{v\in \In(u)} \frac{1}{|\Out(v)|}\right) -\sum_{u:\SelfLoop(u)=1}\left(\frac{1}{|\Out(u)|}\right)\\[2ex]
 \geq & \sum_u \Temp(u) - \sum_u\frac{1}{2}= n- \frac{n}{2} = \frac{n}{2}.
 \numberthis\label{eq:tempprime}
\end{align*}

Similarly to the proof of \cref{them:loopless_bounded_temp}, the fixation probability given that a mutant is initially placed on vertex $u$ is at most
\[
\FixProb(G_n,r,u) \leq 
%\frac{1-\left(\frac{r}{\Temp'(u)}\right)^{-1}}{1-\left(\frac{r}{\Temp'(u)}\right)^{-2}} = \frac{1}{1+\frac{\Temp'(u)}{r}} = 
1-\frac{\Temp'(u)}{\Temp'(u)+r}
\]
Summing over all $u$, we obtain
\[
\FixProb(G_n,r,\Temp)= \frac{1}{n}\cdot \sum_{u} \Temp(u)\cdot \FixProb(G_n,r,u)  \leq  \frac{1}{n}\cdot \sum_{u} \Temp(u)\cdot \left(1-\frac{\Temp'(u)}{\Temp'(u)+r}\right) \leq  1-\frac{1}{n}\cdot \sum_u \frac{\Temp'(u)^2}{\Temp'(u)+r}\ ;
\numberthis \label{eq:temp_upperbound2}
\]
since $\sum_u\Temp(u)=n$ and $\Temp(u)\geq \Temp'(u)$.

Using the Cauchy-Schwarz inequality we get
\[
\sum_u \frac{\Temp'(u)^2}{\Temp'(u)+r}\geq \frac{\left(\sum_u\Temp'(u)\right)^2}{\sum_u (\Temp'(u) + r)} = \frac{x^2}{x+n\cdot r}, 
\]
where $x=\sum_u\Temp'(u)$. Note that the function $f(x)=\frac{x^2}{x+n\cdot r}$ is increasing in $x$ for $x>0$ and any $r,n>0$. Since $x>n/2$, the right-hand side is minimized for $x=n/2$, that is
\[
\sum_u \frac{\Temp'(u)^2}{\Temp'(u)+r}\geq \frac{(n/2)^2}{n/2+n\cdot r} = \frac{n}{4r+2}.
\]
Thus \cref{eq:temp_upperbound2} becomes
\[
\FixProb(G_n,r,\Temp) \leq 1-\frac{1}{n}\cdot \frac{n}{4r+2} = 1-\frac{1}{4r+2}
\]
as desired.

%Using the Cauchy-Schwarz inequality and \cref{eq:tempprime}, we obtain
%\[
%\sum_u \frac{\Temp'(u)^2}{\Temp'(u)+r}\geq \frac{\left(\sum_u\Temp'(u)\right)^2}{\sum_u (\Temp'(u) + r)} \geq \frac{(n/2)^2}{n+n\cdot r} = \frac{n}{4\cdot (r+1)}
%\]
%and thus \cref{eq:temp_upperbound2}  becomes
%\[
%\FixProb(G_n,r,\Temp) \leq 1-\frac{1}{n}\cdot \frac{n}{4\cdot (r+1)} = 1-\frac{1}{4\cdot (r+1)}\ ;
%\]
%as desired.

\end{proof}
\fi

We thus arrive at the following corollary.
\smallskip
\begin{corollary}\label{cor:unweighted_bounded_temp}
There exists no unweighted family of graphs which is a strong temperature amplifier.
\end{corollary}

\subsection{Negative Result~3}\label{subsec:bounded_loopless_uniform}
We now prove our negative result~3. 
%%of the main text.
\smallskip
\begin{theorem}\label{them:loopless_bounded_uniform}
For all self-loop free graphs $G_n$ with $c=\deg(G_n)$, and 
for every $r\geq 1$ we have $\FixProb(G_n,r,\Uniform)\leq 1-1/(c+r\cdot c^2)$.
\end{theorem}
\ifshortproofs
\begin{proof}
Let $G_n=(V,E,W)$ and $\gamma=1/c$. 
Let $V^{h}$ be the set of vertices $u$ which have a neighbor $v \in \In(u)$ such that $W[v,u] \geq \gamma$. 
In words, the set $V^{h}$ contains ``hot'' vertices, since each vertex  $u\in V^h$ is replaced frequently 
(with rate at least $\gamma$) by at least one neighbor $v$.
Consider that the initial mutant is placed on a vertex $u\in V^h$. 
Then using Lemma~\ref{lemm:technical} a simple calculation shows that the fixation probability is at most 
%%the probability that the mutant spreads to two vertices before it goes extinct.
%%The probability of this event is at most 
$(r\cdot c)/(1+r\cdot c)$, and hence it is upperbounded by a constant.
Thus for vertices $u \in V^h$ we have $\FixProb(G_n,r,u) \leq (r\cdot c)/(1+r\cdot c)$.
Additionally, we have $|V^{h}| \geq \frac{n}{c}$, i.e., the set $V^{h}$ contains at least a constant fraction of the vertices. 
Hence the initial mutant will be placed on some hot vertex from $V^{h}$ with at least some constant probability, by uniform
initialization.
To favor fixation even more, we consider that if the initial mutant is placed on some vertex outside $V^{h}$, 
then fixation is achieved with probability~1.
We have 
\[
\FixProb(G_n,r,\Uniform) \leq \frac{|V^{h}|}{n} \cdot \frac{r\cdot c}{1+r\cdot c}  + 
\frac{n-|V^{\gamma|}}{n}\cdot 1 = \frac{1}{c}\cdot \frac{r\cdot c}{1+r\cdot c} + \frac{c-1}{c}   \leq 1-\frac{1}{c+r\cdot c^2}\ .
\]
%For detailed calculation see~\cite{FullReport}.
The desired result follows.
\end{proof}
\fi
\iffullproofs
\begin{proof}
Let $G_n=(V,E,W)$ and $\gamma=1/c$.
For a vertex $u$, denote by $\Out^{\gamma}(u)=\{v\in \Out(u) \ :\ W[u,v] \geq \gamma \}$.
Observe that since $\deg(G_n)=c$, every vertex $u$ has an outgoing edge of weight at least $1/c$, and thus $\Out^{\gamma}(u)\neq \emptyset$ for all $u\in V$.
Let $V^{h}=\bigcup_u \Out^{\gamma}(u)$.
Intuitively, the set $V^{h}$ contains ``hot'' vertices, since each vertex  $u\in V^h$ is replaced frequently 
(with rate at least $\gamma$) by at least one neighbor $v$.

\smallskip\noindent{\em Bound on size of $V^h$.}
We first obtain a bound on the size of $V^h$. 
Consider a vertex $u \in V$ and a vertex $v \in \Out^{\gamma}(u)$ (i.e., $v \in V^h$).
For every vertex $w \in \In(v)$ such that $v\in \Out^{\gamma}(w)$ we can count $v \in V^h$ 
and to avoide multiple counting, we consider for each count of $v$ a contribution of 
$\frac{1}{|\{w\in \In(v):~v\in \Out^{\gamma}(w)\}|}$, which is at least $\frac{1}{c}$ due to
the degree bound.
Hence we have 
\[
|V^{h}| = \sum_{u\in V} \sum_{v\in \Out^{\gamma}(u)} \frac{1}{|\{w\in \In(v):~v\in \Out^{\gamma}(w)\}|} \geq \sum_{u\in V} \sum_{v\in \Out^{\gamma}(u)} \frac{1}{c} \geq \sum_{u\in V} \frac{1}{c}= \frac{n}{c}\ ;
\]
where the last inequality follows from the fact that $\Out^{\gamma}(u)\neq \emptyset$ for all $u\in V$.
Hence the probability that the initial mutant is a vertex in $V^h$ has probability at least $1/c$ according to the uniform
initialization.

\smallskip\noindent{\em Bound on probability.}
Consider that the initial mutant is a vertex $u\in V^{h}$.
%%We upperbound the fixation probability $\FixProb(G_n,r,u)$ starting from $u$ by the probability that a configuration $\State_t$ is reached with $|\State_t|=2$.
Consider any configuration $\State_i=\{u\}$, we have the following:
\[
x=\Prb^{G_n,r}[|\State_{i+1}|=2 \mid \State_i=\{u\}] = \frac{r}{\Fitness(\State_i)}
\] 
%%\qquad  \text{ and } \qquad \Probr{|\State_{i+1}|=0}=
\[
y=\Prb^{G_n,r}[|\State_{i+1}|=0 \mid \State_i=\{u\}] =
\frac{1}{\Fitness(\State_i)}\cdot \sum_{(v,u)\in E} W[v,u] 
\geq \frac{1}{\Fitness(\State_i)}\cdot \sum_{v:u\in \Out^{\gamma}(v)} \gamma
\geq \frac{1}{\Fitness(\State_i)}\cdot \gamma\ .
\]
Thus $x/y \leq r/\gamma$.
Hence by Lemma~\ref{lemm:technical} we have
\[
\FixProb(G_n,r,u) \leq 
%\frac{1-\left(\frac{r}{\gamma}\right)^{-1}}{1-\left(\frac{r}{\gamma}\right)^{-2}} = \frac{1}{1+\frac{\gamma}{r}}= \frac{1}{1+\frac{1}{r\cdot c}} = 
\frac{r\cdot c}{1+r\cdot c}\ .
\]

Finally, we have
\begin{align*}
\FixProb(G_n,r,\Uniform)= & 
\sum_{u\in V^{h}} \frac{1}{n}\cdot \FixProb(G_n,r,u)
+ \sum_{u\in V\setminus V^{h}} 
\frac{1}{n}\cdot \FixProb(G_n,r,u) \\[2ex]
\leq &  \frac{1}{c}\cdot \frac{r\cdot c}{1+r\cdot c} + \frac{c-1}{c}\cdot 1 = 1-\frac{1}{c}\cdot \left(1-\frac{r\cdot c}{1+r\cdot c}\right) = 1-\frac{1}{c+r\cdot c^2} \ .
\end{align*}
The desired result follows.
\end{proof}
\fi

We thus arrive at the following corollary.
\smallskip
\begin{corollary}\label{cor:loopless_bounded_uniform}
There exists no self-loop free, bounded-degree family of graphs which is a strong uniform amplifier.
\end{corollary}

\subsection{Negative Result~4}\label{subsec:bounded_weightless_uniform}
We now prove our negative result~4. %%of the main text.
\smallskip
\begin{theorem}\label{them:unweighted_bounded_uniform}
For all unweighted graphs $G_n$ with $c=\deg(G_n)$, and 
for every $r\geq 1$ we have  $\FixProb(G_n,r,\Uniform)\leq 1-1/(1+r\cdot c)$.
\end{theorem}
\ifshortproofs
\begin{proof}
Let $G_n=(V,E,W)$ and consider that the initial mutant is placed on a vertex $u$. 
Then the fixation probability is at most 
%the probability that the mutant spreads to two vertices before it goes extinct.
%%The probability of this event is at most 
$r\cdot c/(1+r\cdot c)$ (using Lemma~\ref{lemm:technical} and simple calculations).
Summing over all vertices $u$, we obtain
\[
\FixProb(G_n,r,\Uniform) = \frac{1}{n}\cdot \sum_u\FixProb(G_n,r,u) \leq \frac{r\cdot c}{1+r\cdot c} = 1-\frac{1}{1+r\cdot c}\ .
\]
%For detailed calculation see~\cite{FullReport}.
The desired result follows.
\end{proof}
\fi
\iffullproofs
\begin{proof}
Let $G_n=(V,E,W)$ and consider that  $\State_0=u$ for some $u\in V$.
%%We upperbound the fixation probability $\FixProb(G_n,r,u)$ starting from $u$ by the probability that a configuration $\State_t$ is reached with $|\State_t|=2$.
Consider any configuration $\State_i=\{u\}$, we have the following:
\[
x=\Prb^{G_n,r}[|\State_{i+1}|=2 \mid \State_i=\{u\}] \leq \frac{r}{\Fitness(\State_i)}\ .
\]
%%\qquad  \text{ and } \qquad \Probr{|\State_{i+1}|=0}=
\[
y=\Prb^{G_n,r}[|\State_{i+1}|=0 \mid \State_i=\{u\}] =
\frac{1}{\Fitness(\State_i)}\cdot \sum_{v\in \In(u)\setminus\{u\}} W[v,u] 
\geq \frac{1}{\Fitness(\State_i)}\cdot \frac{1}{c}\ .
\]
Thus $x/y \leq  r\cdot c$.
By Lemma~\ref{lemm:technical} we have
\[
\FixProb(G_n,r,u)\leq %\frac{1-\left(r\cdot c \right)^{-1}}{1-\left(r\cdot c \right)^{-2}} =  
\frac{r\cdot c}{1+r\cdot c}.
\]
Finally, we have
\begin{align*}
\FixProb(G_n,r,\Uniform) =\frac{1}{n}\cdot \sum_u \FixProb(G_n,r,u) \leq \frac{r\cdot c}{1+r\cdot c} = 1-\frac{1}{1+r\cdot c}\ .
\end{align*}
The desired result follows.
\end{proof}
\fi

We thus arrive at the following corollary.
\smallskip
\begin{corollary}\label{cor:weightless_bounded_uniform}
There exists no unweighted, bounded-degree family of graphs which is a strong uniform amplifier.
\end{corollary}

\smallskip
\begin{remark}\label{rem:bounded_degree_uniform}
\cref{them:loopless_bounded_uniform,them:unweighted_bounded_uniform} establish the nonexistence of strong amplification with bounded degree graphs.
A relevant result can be found in~\cite{Mertzios13}, which establishes an upperbound of the fixation probability of
mutants under uniform initialization on unweighted, undirected graphs.
If the bounded degree restriction is relaxed to bounded average degree, then recent results
show that strong amplifiers (called \emph{sparse incubators}) exist~\cite{Goldberg17}.
%%% CITE B,C
\end{remark}

%% file: three_state_rw.tex
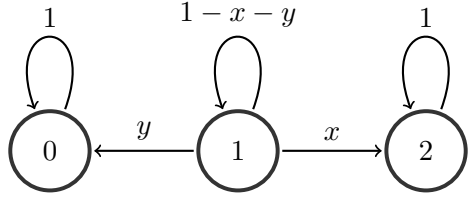
\begin{figure}[!h]
\centering
\begin{tikzpicture}[thick,
pre/.style={<-,shorten >= 1pt, shorten <=1pt, thick,  bend angle = 20},
post/.style={->,shorten >= 1pt, shorten <=1pt,  thick,  bend angle = 20},
every loop/.style={<-,shorten >= 1pt, shorten <=1pt, thick, auto, in=110,out=70, looseness=8, min distance=14mm},
state/.style={circle,draw=black!80, line width=1.5pt, inner sep=2pt,minimum size=30pt},
target state/.style={circle, draw=black!80, line width = 1.2pt, minimum size = 25pt},
virt/.style={circle,draw=white!50,fill=white!20,thick, inner sep=2pt,minimum size=30pt},
node distance=2.5cm]

\newcommand{\xdisp}{2.7}
\newcommand{\ydisp}{1.5}
\newcommand{\bend}{20}
\newcommand{\bendtwo}{45}

\node [state] (0) at (0,0)	{$0$};
\node [state, right of=0]	(1)	{$1$};
\node [state, right of=1]	(2)	{$2$};

\draw[post,] (1) to node[auto, above]	{$y$} (0);
\draw[post,] (1) to node[auto, above]	{$x$} (2);

\draw[post, loop above] (1) to node[auto, above]	{$1-x-y$} (0);
\draw[post, loop above] (0) to node[auto, above]	{$1$} (0);
\draw[post, loop above] (2) to node[auto, above]	{$1$} (2);

\end{tikzpicture}
\caption{Illustration of the Markov chain of \cref{lemm:technical}.}
\label{fig:three_state_rw}
\end{figure}

%% file: amplifiers.tex
\section{Positive Result}\label{sec:amplifiers}

In the previous section we showed that self-loops and weights are necessary for the existence 
of strong amplifiers.
In this section we present our positive result, namely that every family of undirected graphs 
with self-loops and whose diameter is not ``too large'' can be made a strong amplifier by using  
appropriate weight functions.
Our result relies on several novel conceptual steps, therefore the proof is structured in three parts.
\begin{enumerate}
\item First, we introduce some formal notation that will help with the exposition of the ideas that follow.
\item Second, we describe an algorithm which takes as input an undirected graph $G_n=(V_n,E_n)$ of $n$ vertices,
and constructs a weight matrix $W_n$ to obtain the weighted graph $G^{\Weight}_n=(V_n, E_n, W_n)$.
\item Lastly, we prove that $G_n^{\Weight}$ is a strong amplifier both for uniform and temperature initialization.
\end{enumerate}
%Below we write whp to denote ``with high probability''.
%\begin{compactenum}
%\item In \cref{subsec:weights} we consider as input an unweighted, undirected graph $G_n=(V_n,E_n)$ of $n$ 
%vertices and diameter $\Diameter(G_n)\leq n^{1-\eps}$, for some fixed $\eps>0$, and construct a weight function.
%Conceptually, the weight function separates the graph into two components of 
%appropriate sizes: (i)~a fast evolving component, which we call 
%the \emph{sink}, and (ii)~a slowly evolving component, which we call 
%the \emph{branches}.
%
%\item In \cref{subsec:sink_evolution} we establish some results related to the 
%evolution in the hub.
%Crucially, we show that if an advantageous mutant ($r>1$) is initially placed 
%outside the hub, then whp. the hub will eventually be occupied by mutants.
%
%\item In \cref{subsec:fixation_prob} we show that starting from a hub of only mutants, whp the branches of the 
%graph will eventually consist of only mutants.
%Combined with the previous item, we establish that under uniform, temperature and $\coeff$-convex initialization, 
%whp the initial mutant fixates in the graph.
%We also show that if instead the mutants are disadvantageous ($r<1$), then whp the initial mutant goes extinct.
%\end{compactenum}

Before presenting the details we introduce some notation to be used in this section.

\subsection{Undirected Graphs and Notation}\label{subsec:notation}
We first present some additional notation required for the exposition of the results of this section.

\smallskip\noindent{\em Undirected graphs.}
Our input is an unweighted undirected graph $G_n=(V_n,E_n)$ with self loops.
For ease of notation, we drop the subscript $n$ and refer to the graph $G=(V,E)$ instead.
Since $G$ is undirected, for all vertices $u$ we have $\In(u)=\Out(u)$, and we denote
by $\Neighbors(u)=\In(u)=\Out(u)$ the set of neighbors of vertex $u$.
Hence, $v\in \Neighbors(u)$ iff $u\in \Neighbors(v)$.
Moreover, since $G$ has self-loops, we have $u\in \Neighbors(u)$.
Also we consider that $G$ is connected, i.e., for every pair of vertices $u,v$, there is a path from $u$ to $v$.

\smallskip\noindent{\em Symmetric weight function.}
So far we have used a stochastic weight matrix $W$, where for every $u$ we have $\sum_{v} W[u,v]=1$.
In this section, we will consider a weight function $\Weight:E\to \Realsplus$, and given a vertex $u \in V$ 
we denote by $\Weight(u)=\sum_{v\in \Neighbors(u)}\Weight(u,v)$.
Our construction will not only assign weights, but also ensure symmetry.
In other words, we we construct \emph{symmetric} weights such that for all $u,v$ we have  
$\Weight(u,v)=\Weight(v,u)$.
Given such a weight function $\Weight$, the corresponding stochastic weight matrix $W$ is defined as
$W[u,v]=\Weight(u,v)/\Weight(u)$ for all pairs of vertices $u,v$.
Given a unweighted graph $G$ and weight function $\Weight$, we denote by $G^{\Weight}$
the corresponding weighted graph.

\smallskip\noindent{\em Vertex-induced subgraphs.}
Given a set of vertices $X\subseteq V$, we denote by $G^{\Weight}\restr{X}=(X, E\restr{X}, \Weight\restr{X})$
the subgraph of $G$ induced by $X$, where $E\restr{X}= E\cap (X\times X)$,
and the weight function $\Weight\restr{X}: E\restr{X}\to \Realsplus$ defined as
\[
\Weight\restr{X}(u,v)=
\left\{
\begin{array}{lr}
\Weight(u,u) + \sum_{(u,w)\in E\setminus E\restr{X} }\Weight(u,w) & \text{ if } u=v\\
\Weight(u,v)&  \text{ otherwise }
\end{array}
\right.
\]
In words, the weights on the edges of $u$ to vertices that do not belong to $X$ are added to the self-loop weight of $u$.
Since the sum of all weights does not change, we have $\Weight\restr{X}(u)=\Weight(u)$ for all $u$.
The temperature of $u$ in $G\restr{X}$ is
\[
\Temp\restr{X}(u) = \sum_{v\in \Neighbors(u)\cap X} \frac{\Weight\restr{X}(v,u)}{ \Weight\restr{X}(v)}\ .
\]

\subsection{Algorithm for Weight Assignment on $G$}\label{subsec:weights}
We start with the construction of the weight function $\Weight$ on $G$.
Since we consider arbitrary input graphs, $\Weight$ is constructed by an algorithm.
The time complexity of the algorithm is $O(n\cdot \log n)$.
Since our focus is on the properties of the resulting weighted graph, we do not explicitly analyze 
the time complexity.

\smallskip\noindent{\bf Steps of the construction.}
Consider a connected graph $G$ with diameter $\Diameter(G)\leq n^{1-\eps}$, where $\eps>0$ is a constant independent of $n$.
We construct a weight function $\Weight$ such that whp an initial mutant arising under uniform or temperature initialization, eventually fixates on $G^{\Weight}$.
The weight assignment consists of the following conceptual steps.

\begin{compactenum}
\item {\em Spanning tree construction and partition.}
First, we construct a \emph{spanning tree} $\SpanTree_n^x$ of $G$ rooted on some arbitrary vertex $x$.
In words, a spanning tree of an undirected graph is a connected subgraph that is a tree and includes 
all of the vertices of the graph.
Then we partition the tree into a number of component trees of appropriate sizes.

\item {\em Hub construction.} Second, we construct the \emph{hub} of $G$, which consists of the vertices $x_i$ that are roots of the component trees, together with all vertices in the paths that connect each $x_i$ to the root $x$ of $\SpanTree_n^x$. All vertices that do not belong to the hub belong to the \emph{branches} of $G$.

\item {\em Weight assignment.} 
Finally, we assign weights to the edges of $G$, such that the following properties hold:
\begin{compactenum}
\item The hub is an isothermal graph, and evolves exponentially faster than the branches.
\item All edges between vertices in different branches are effectively cut-out (by being assigned weight $0$).
\end{compactenum}
\end{compactenum}

In the following we describe the above steps formally.

\smallskip\noindent{\bf Spanning tree $\SpanTree_n^x$ construction and partition.}
Given the graph $G$, we first construct a spanning tree using the standard breadth-first-search (BFS) algorithm.
Let $\SpanTree_n^x$ be such a spanning tree of $G$, rooted at some arbitrary vertex $x$.
We now construct the partitioning as follows:
We choose a constant $c=2\eps/3$, and pick a set $S\subset V$ such that
\begin{compactenum}
\item $|S|\leq n^c$, and
\item  the removal of $S$ splits $\SpanTree_n^x$ into $k$ trees $T^{x_1}_{n_1},\dots, T^{x_k}_{n_k}$, each $T^{x_i}_{n_i}$ rooted at vertex $x_i$ and of size $n_i$, with the property that $n_i\leq n^{1-c}$ for all $1\leq i\leq k$.
\end{compactenum}
The set $S$ is constructed by a simple bottom-up traversal of $\SpanTree_n^x$ in which we keep track of the size $\Size(u)$
of the subtree marked by the current vertex $u$ and the vertices already in $S$. Once $\Size(u)> n^{1-c}$, we add $u$ to $S$ and proceed as before.
Since every time we add a vertex $u$ to $S$ we have $\Size(u)>n^{1-c}$, it follows that $|S|\leq n^c$.
Additionally, the subtree rooted in every child of $u$ has size at most $n^{1-c}$, otherwise that child of $u$ would have been chosen to be included in $S$ instead of $u$.

\smallskip\noindent{\bf Hub construction: hub $\Sink$.}
Given the set of vertices $S$ constructed during the spanning tree partitioning, we construct the set of vertices $\Sink\subset V$ called the \emph{hub}, as follows:
\begin{compactenum}
\item We choose a constant $\gamma=\eps/3$.
\item For every vertex $u\in S$, we add in $\Sink$ every vertex $v$ that lies in the unique simple path $P_u:x\Path u$ between the root $x$ of $\SpanTree_n^x$ and $u$ (including $x$ and $u$). Since $\Diameter(G)\leq n^{1-\eps}$ and $|S|\leq n^c$, we have that $|\Sink|\leq n^{1-\eps+c} \leq n^{1-\gamma}$.
\item We add $n^{1-\gamma}-|\Sink|$ extra vertices to $\Sink$, such that in the end, the vertices of $\Sink$ form a connected subtree of $\SpanTree_n^x$ (rooted in $x$).
This is simply done by choosing a vertex $u\in \Sink$ and a neighbor $v$ of $u$ with $v\not \in \Sink$, and adding $v$ to $\Sink$, until $\Sink$ contains $n^{1-\gamma}$ vertices.
\end{compactenum}

\smallskip\noindent{\bf Branches $B_j=T_{m_j}^{y_j}$.}
The hub $\Sink$ defines a number of trees $B_j=T_{m_j}^{y_j}$, where each tree is rooted at a vertex $y_j\not\in \Sink$ adjacent to $\Sink$, and has $m_j$ vertices. 
We will refer to these trees as \emph{branches}(see~\cref{fig:sink_branches}).

\smallskip
\begin{proposition}\label{prop:sizes}
Note that by construction, we have $m_j\leq n^{1-2/3\cdot \eps}$ for every $j$, and $|\Sink|=n^{1-\eps/3}$, and 
$\sum_j m_j = n-n^{1-\eps/3}$.
\end{proposition}

%\begin{remark}\label{rem:sizes}
%By \cref{prop:sizes}, the size of each branch is much smaller than the size of the hub;
%but the sum of the sizes of all branches is much larger than the size of the hub.
%\end{remark}

\begin{figure}
\centering
\includegraphics[scale=1]{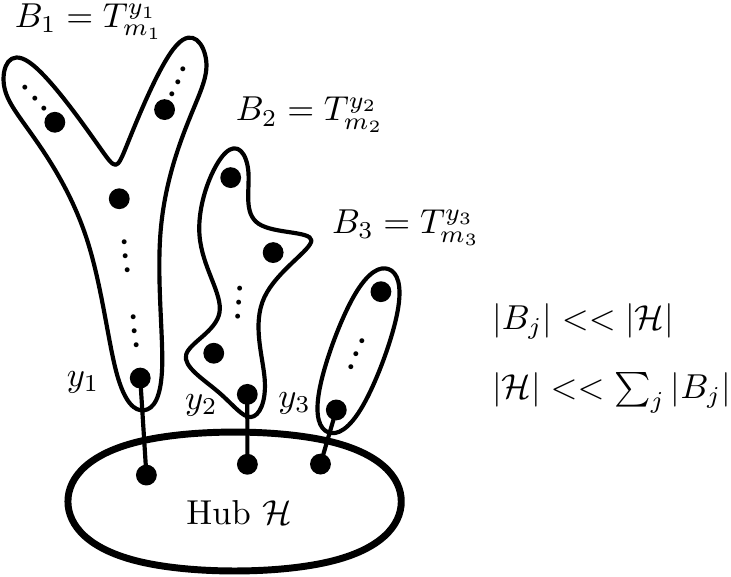}
\caption{
Illustration of the hub $\Sink$ and the branches $T_{m_j}^{y_j}$.
}
\label{fig:sink_branches}
\end{figure}

\smallskip\noindent{\bf Notation.} 
To make the exposition of the ideas clear, we rely on the following notation.
\begin{compactenum}
\item {\em Parent $\Parent(u)$ and ancestors $\Ancestors(u)$.}
Given a vertex $u\neq x$, we denote by $\Parent(u)$ the parent of $u$ in $\SpanTree_n^x$ and by $\Ancestors(u)$ the set of ancestors of $u$.

\item {\em Children $\Children(u)$ and descendants $\Descendants(u)$.} 
Given a vertex $u$ that is not a leaf in $\SpanTree_n^x$, we denote by $\Children(u)$ the children of $u$  in $\SpanTree_n^x$ that do not belong to the hub $\Sink$,
and by $\Descendants(u)$ the set of descendants of $u$ in $\SpanTree_n^x$ that do not belong to the hub $\Sink$.
\end{compactenum}

\smallskip\noindent{\bf Frontier, distance, and branches.} We present few notions required for the weight assignment:
\begin{compactenum}
\item {\em Frontier $\Frontier$.}
Given the hub $\Sink$, the \emph{frontier} of $\Sink$ is the set of vertices $\Frontier\subseteq \Sink$ defined as 
\[
\Frontier = \bigcup_{u\in V\setminus\Sink} \Neighbors(u) \cap \Sink\ .
\]
In words, $\Frontier$ contains all vertices of $\Sink$ that have a neighbor not in $\Sink$.

\item {\em Distance function $\Distance$.}
For every vertex $u$, we define its \emph{distance} $\Distance(u)$ to be the length of the shortest path $P:u\Path v$ in $T_n^{x}$
to some vertex $v\in \Frontier$
(e.g., if $u\in \Frontier$, we have (i)~$\Distance(u)=0$, and (ii)~for every $v\in \Neighbors(u)\setminus \Sink$ we have $\Distance(v)=1$).

%\item {\em Branches.} The hub $\Sink$ defines a number of trees $T_{m_j}^{y_j}$, where each tree is rooted at a vertex $y_j$ with $\Distance(y_j)=1$ (i.e., the parent of $y_j$ belongs to the hub), and has $m_j$ vertices. 
%Note that by construction, we have $m_j\leq n^{1-c}$ for every $j$.
%We will refer to these trees as \emph{branches}.

\item {\em Values $\mu$ and $\nu$.}
For every vertex $u\in \Sink$, we define $\Degree(u)=|(\Neighbors(u)\cap \Sink)\setminus\{u\}|$
i.e., $\Degree(u)$ is the number of neighbors of $u$ that belong to the hub (excluding $u$ itself).
Let
\[
\mu=\max_{u\in \Frontier} |\Children(u)| \qquad  \text{ and } \qquad \nu=\max_{u\in \Sink}\Degree(u)\ .
\]
\end{compactenum}

\smallskip\noindent{\bf Weight assignment.}
We are now ready to define the weight function $\Weight:E\to \Realsplus$.

\begin{compactenum}
\item For every edge $(u,v)$ such that $u\neq v$ and $u,v\not \in \Sink$ and $u$ and $v$ are not neighbors in $\SpanTree_n^x$, we assign $\Weight(u,v)=0$.
\item For every vertex $u\in \Frontier$ we assign $\Weight(u,u)=(\mu-|\Children(u)|)\cdot 2^{-n}+\nu-\Degree(u)$.
\item For every vertex $u\in \Sink\setminus\Frontier$ we assign $\Weight(u,u)=\mu\cdot 2^{-n}+\nu-\Degree(u)$.
\item For every vertex $u\not\in\Sink$ we assign $\Weight(u,u)=n^{-2\cdot\Distance(u)}$.
\item For every edge $(u,v)\in E$ such that $u\neq v$ and $u,v\in \Sink$ we  assign $\Weight(u,v)=1$.
\item For every remaining edge $(u,v)\in E$ such that $u=\Parent(v)$ we assign $\Weight(u,v)=2^{-n}\cdot n^{-4\cdot \Distance(u)}$.
\end{compactenum}

%\cref{fig:weights} illustrates the assignment of the symmetric weights $\Weight$ on the Star graph.
%
%\begin{figure}
%\centering
%\includegraphics[scale=1]{star_abstract.pdf}
%\caption{
%Illustration of the weight assignment on a Star graph of $6$ vertices. The spanning tree is rooted on the central vertex of the Star,
%and the hub consists of the central vertex together with two leaf vertices.
%}
%\label{fig:weights}
%\end{figure}

The following lemma is straightforward from the weight assignment, and captures that every vertex in the hub has the same weight.
%\ifshortproofs
%We refer to the full version~\cite{FullReport} for the detailed proof.
%\fi
\smallskip
\begin{lemma}\label{lem:sink_weights}
For every vertex $u\in \Sink$ we have $\Weight(u)=\sum_{v\in\Neighbors(u)}\Weight(u,v)=\mu\cdot 2^{-n}+\nu$.
\end{lemma}
\iffullproofs
\begin{proof}
Consider any vertex $u\in \Sink\setminus\Frontier$. We have
\begin{align*}
\Weight(u) =& \Weight(u,u) + \sum_{v\in\Neighbors(u)\setminus\{u\}} \Weight(u,v)\\
=& \mu\cdot 2^{-n}+\nu - \Degree(u) + \sum_{v\in\Neighbors(u)\setminus\{u\}} 1\\
=& \mu\cdot 2^{-n}+\nu - \Degree(u) + \Degree(u) \\
=& \mu\cdot 2^{-n}+\nu
\numberthis\label{eq:weight_sink}
\end{align*}

Similarly, consider any $u\in \Frontier$. We have 
\begin{align*}
\Weight(u) =& \Weight(u,u) + \sum_{v\in(\Neighbors(u)\cap \Sink)\setminus\{u\}} \Weight(u,v)+ \sum_{v\in \Children(u)} \Weight(u,v)\\
=& (\mu-|\Children(u)|)\cdot 2^{-n}+\nu-\Degree(u) + \sum_{v\in(\Neighbors(u)\cap \Sink) \setminus\{u\}} 1 + \sum_{v\in \Children(u)}2^{-n} \\
=& \mu\cdot 2^{-n} - |\Children(u)|\cdot 2^{-n} +\nu - \Degree(u) + \Degree(u)  + |\Children(u)|\cdot 2^{-n}\\
=& \mu\cdot 2^{-n}+\nu
\numberthis\label{eq:weight_frontier}
\end{align*}
\end{proof}
\fi

\subsection{Analysis of the Fixation Probability}
In this section we present detailed analysis of the fixation probability and we start with the
outline of the proof.

\subsubsection{Outline of the proof}
%Recall Fig.~2 of the main article which illustrates the four stages to fixation, which we outline here.
The fixation of new mutants is guaranteed by showing that each of the following four stages happens with high probability.
\begin{compactenum}
\item[(A)] In stage~1 we consider the event $\ev_1$ that a mutant arises in one of the branches  (i.e., outside the hub $\Sink$).
We show that event $\ev_1$ happens whp.

\item[(B)] In stage~2 we consider the event $\ev_2$ that a mutant occupies a vertex $v$ of the branches 
which is a neighbor to the hub. We show that given event $\ev_1$ the event $\ev_2$ happens whp.

\item[(C)] In stage~3 we consider the event $\ev_3$ that the mutants fixate in the hub.
We show that given event $\ev_2$ the event $\ev_3$ happens whp.

\item[(D)] In stage~4 we consider the event $\ev_4$ that the mutants fixate in all the branches.
We show that given event $\ev_3$ the event $\ev_4$ happens whp.

\end{compactenum}

\smallskip\noindent{\bf Crux of the proof.} 
Before the details of the proof we present the main crux of the proof. 
We say a vertex $v \not \in \Sink$ hits the hub when it places an offspring to the hub.
First, our construction ensures that the hub is isothermal. 
Second, our construction ensures that a mutant appearing in a branch reaches to a vertex adjacent
to the hub, and hits the hub with a mutant polynomially many times.
Third, our construction also ensures that the hub reaches a homogeneous configuration whp 
between any two hits to the hub. 
We now describe two crucial events. 
\begin{itemize}
\item Consider that a mutant is adjacent to a hub of residents.
Every time a mutant is introduced in the hub it has a constant probability (around $1-1/r$ for large population)
of fixation since the hub is isothermal. 
The polynomially many hits of the hub by mutants ensure that the hub becomes mutants whp. 

\item In contrast consider that a resident is adjacent to a hub.
Every time a resident is introduced in the hub it has exponentially small probability 
(around $(r-1)/(r^{|\Sink|}-1)$) of fixation.
\end{itemize}
Hence, given a hub of mutants, the probability (say, $\eta_1=2^{-\Omega(|\Sink|)}$) that the residents win over the 
hub is exponentially small.
Given a hub of mutant, the probability that the hub wins over a branch $B_j$ is also 
exponentially small  (say, $\eta_2=2^{-O(|B_j|)}$). 
More importantly the ratio of $\eta_1/\eta_2$ is also exponentially small 
(by \cref{prop:sizes} regarding the sizes of the hub and branches).
Using this property, se show that fixation the mutants reach fixation whp.
We now analyze each stage in detail.

\subsubsection{Analysis of Stage~1: Event $\ev_1$}

\begin{lemma}\label{lem:initialization}
Consider the event $\ev_1$ that the initial mutant is placed at a vertex outside the hub.
Formally, the event $\ev_1$ is that $\State_0 \cap \Sink=\emptyset$.
The event $\ev_1$ happens with probability at least $1-O(n^{-\eps/3})$,
i.e., the event $\ev_1$ happens whp.
\end{lemma}
\ifshortproofs
\begin{proof}
We present the result for the uniform and temperature initialization below.
\begin{compactitem}
\item \emph{(Uniform initialization):} The initial mutant is placed on a vertex $u\not\in \Sink$ with probability
\[
\sum_{u\not\in \Sink}\frac{1}{n}=\frac{|V\setminus \Sink|}{n} = \frac{n-n^{1-\gamma}}{n}  = 1-\frac{n^{1-\gamma}}{n}=1-O(n^{-\eps/3})\ ;
\]
since by definition $\gamma=\eps/3$.

\item \emph{(Temperature initialization):} For every vertex $u\not\in \Sink$, we have
\[
\sum_{v\in \Neighbors(u)\setminus\{u\}}\Weight(v,u) \leq  \sum_{v\in \Neighbors(u)\setminus\{u\}} 2^{-n} = 2^{-\Omega(n)}\ ;
\]
whereas since $\Diameter(G)\leq n^{1-\eps}$ we have
\[
\Weight(u,u) = n^{-2\cdot \Distance(u)} \geq n^{-2\cdot \Diameter(G)} \geq  n^{-O(n^{1-\eps})}\ .
\]
Let $A=\Weight(u,u)$ and $B=\sum_{v\in \Neighbors(u)\setminus\{u\}}\Weight(v,u)$, and we have
\[
\frac{\Weight(u,u)}{\Weight(u)} = \frac{A}{A + B} = 1-\frac{B}{A+B} = 1-2^{-\Omega(n)}\ .
\]
Then the desired event happens with probability at least
\begin{align*}
 \sum_{u\not \in \Sink} \frac{\Temp(u)}{n}
= &\frac{1}{n} \cdot \sum_{u\not \in \Sink} \sum_{v\in \Neighbors(u)} \frac{\Weight(v,u)}{ \Weight(v)} 
\geq  \frac{1}{n} \cdot \sum_{u\not \in \Sink} \frac{\Weight(u,u)}{\Weight(u)} 
\geq \frac{1}{n} \cdot \sum_{u\not \in \Sink} \left(1-2^{-\Omega(n)}\right) \\
=& 1-O(n^{-\eps/3})
\end{align*}
since $\gamma=\eps/3$. The desired result follows.
\end{compactitem}
\end{proof}
\fi

\iffullproofs
\begin{proof}
We examine the uniform and temperature initialization schemes separately.
\begin{compactitem}
\item \emph{(Uniform initialization):} The initial mutant is placed on a vertex $u\not\in \Sink$ with probability
\[
\sum_{u\not\in \Sink}\frac{1}{n}=\frac{|V\setminus \Sink|}{n} = \frac{n-n^{1-\gamma}}{n}  = 1-\frac{n^{1-\gamma}}{n}=1-O(n^{-\eps/3})\ ;
\]
since $\gamma=\eps/3$.

\item \emph{(Temperature initialization):} For any vertex $u\not\in \Sink$, we have
\[
\sum_{v\in \Neighbors(u)\setminus\{u\}}\Weight(u,v) \leq  \sum_{v\in \Neighbors(u)\setminus\{u\}} 2^{-n} = 2^{-\Omega(n)}\ ;
\]
whereas since $\Diameter(G)\leq n^{1-\eps}$ we have
\[
\Weight(u,u) = n^{-2\cdot \Distance(u)} \geq n^{-2\cdot \Diameter(G)} \geq  n^{-O(n^{1-\eps})}\ .
\]
Note that
\[
n^{-O(n^{1-\eps})}=2^{-O(n^{1-\eps}\cdot \log n)} >> 2^{-O(n)}\ .
\]
Let $A=\Weight(u,u)$ and $B=\sum_{v\in \Neighbors(u)\setminus\{u\}}\Weight(u,v)$, and we have
\[
\frac{\Weight(u,u)}{\Weight(u)} = \frac{A}{A + B} = 1-\frac{B}{A+B} = 1-\frac{2^{-\Omega(n)}}{n^{-O(n^{1-\eps})}+2^{-\Omega(n)}}
= 1-\frac{2^{-\Omega(n)}}{n^{-O(n^{1-\eps})}} = 1-2^{-\Omega(n)}\ .
\]
Then the desired event happens with probability at least
\begin{align*}
\sum_{u\not\in \Sink}\ProbT{\State_0=\{u\}} 
=& \sum_{u\not \in \Sink} \frac{\Temp(u)}{n}
= \frac{1}{n} \cdot \sum_{u\not \in \Sink} \sum_{v\in \Neighbors(u)} \frac{\Weight(u,v)}{ \Weight(v)} 
\geq  \frac{1}{n} \cdot \sum_{u\not \in \Sink} \frac{\Weight(u,u)}{\Weight(u)} 
\geq \frac{1}{n} \cdot \sum_{u\not \in \Sink} \left(1-2^{-\Omega(n)}\right) \\
=& \frac{|V\setminus \Sink|}{n}\cdot \left(1-2^{-\Omega(n)}\right)
=  \frac{n-n^{1-\gamma}}{n}\cdot \left(1-2^{-\Omega(n)}\right)  = (1-n^{-\gamma})\cdot  \left(1-2^{-\Omega(n)}\right)\\
=& 1-O(n^{-\eps/3})
\end{align*}
since $\gamma=\eps/3$. The desired result follows.

\end{compactitem}
\end{proof}
\fi

\subsubsection{Analysis of Stage~2: Event $\ev_2$}

The following lemma states that if a mutant is placed on a vertex $w$ outside the hub, then whp
the mutant will propagate to the ancestor $v$ of $w$ at distance $\Distance(v)=1$ from the hub
(i.e., the parent of $v$ belongs to the hub).
This is a direct consequence of the weight assignment, which guarantees that for every vertex $u\not\in\Sink$,
the individual occupying $u$ will place an offspring on the parent of $u$ before some neighbor of $u$ places an offspring on $u$,
and this event happens with probability at least $1-O(n^{-1})$.

\smallskip
\begin{lemma}\label{lem:ancestor_mutant}
Consider that at some time $j$ the configuration of the Moran process on $G^{\Weight}$ is $\State_j=\{w\}$ with $w\not\in \Sink$.
Let $v\in \Ancestors(w)$ with $\Distance(v)=1$, i.e., $v$ is the ancestor of $w$ and $v$ is adjacent to the hub.
Then a subsequent configuration $\State_{t}$ with $v\in \State_t$ is reached with probability $1-O(n^{-1})$, i.e., 
given event $\ev_1$, the event $\ev_2$ happens whp.
\end{lemma}

\ifshortproofs
\begin{proof}
Consider any configuration $\State_i$, with $i\geq j$, and let $u\neq v$ be the highest ancestor of $w$ that is occupied by a mutant (initially it is $u=w$).
Let $\rho_1$ be the probability that $u$ reproduces and places an offspring on its parent $\Parent(u)$, and $\rho_2$ be the probability that a neighbor of $u$ places an offspring on $u$.
It is straightforward from the weight assignment that
\[
\frac{\rho_1}{\rho_2} = \Omega(n)\ ;
\]
that is, whp $u$ places an offspring on its parent before $u$ is replaced by a neighbor.
Hence, if $\Walk=s_i, s_{i+1},\dots $ is the one-dimensional random walk that tracks the highest ancestor of $u$ which is occupied by a mutant,
the forward bias of $\Walk$ is $\Omega(n)$. 
The desired result follows directly from the analysis of biased random walks (see, e.g.,~\cite{Kemeny12},~\cite[Section~6.3]{Nowak06b}).
\end{proof}
\fi
\iffullproofs
\begin{proof}
Let $t$ be the first time such that $v\in \State_t$ (possibly $t=\infty$, denoting that $v$ never becomes mutant). Let $s_i$ be the random variable such that
\[
s_i=
\left\{
\begin{array}{lr}
|\State_i\cap \Ancestors(w)| &\text{ if } i<t\\
| \Ancestors(w)| &\text{ if } i\geq t
\end{array}
\right.
\]
In words, $s_i$ counts the number of mutant ancestors of $u$ until time $t$.
Given the current configuration $\State_i$ with  $0<s_i<|\Ancestors(w)|$, let $u=\arg\min_{z\in \State_i\cap \Ancestors(w)} \Distance(z)$.
The probability that $s_{i+1}=s_i+1$ is lowerbounded by the probability that $u$ reproduces and places an offspring on $\Parent(u)$.
Similarly, the probability that $s_{i+1}=s_i-1$ is upperbounded by the probability that (i)~$\Parent(u)$ reproduces and places an offspring on $u$, plus (ii)~the probability that some $z\in\Descendants(u)\setminus\State_i$ reproduces and places an offspring on $\Parent(z)$.

We now proceed to compute the above probabilities.
Consider any configuration $\State_i$, and 
%vertex $u \in \State_i\setminus\Sink$.Let 
and let $z$ be any child of $u$ and $z'$ any child of $z$.
The above probabilities crucially depend on the following quantities:
\[
\frac{\Weight(u, \Parent(u))}{\Weight(u)};\qquad  \frac{\Weight(u, \Parent(u))}{\Weight(\Parent(u))}; \qquad \sum_{z_i\in \Descendants(u)}\frac{\Weight(\Parent(z_i),z_i)}{\Weight(z_i)}\ .
\]

Recall that
\begin{compactitem}
\item $\Weight(u, \Parent(u))=2^{-n}\cdot n^{-4\cdot \Distance(\Parent(u))}$
\item $\Weight(u,x)=2^{-n}\cdot n^{-4\cdot \Distance(u)}$
\item $\Weight(z,z')=2^{-n}\cdot n^{-4\cdot \Distance(z)}$
\item $\Weight(\Parent(u), \Parent(\Parent(u))) = 2^{-n}\cdot n^{-4\cdot \Distance(\Parent(\Parent(u)))}$
\item $\Weight(u,u)=n^{-2\cdot\Distance(u)}$
\item $\Weight(\Parent(u), \Parent(u))=n^{-2\cdot\Distance(\Parent(u))}$
\item $\Weight(z,z)=n^{-2\cdot\Distance(z)}$
\end{compactitem}

Thus, we have

\begin{align*}
\frac{\Weight(u, \Parent(u))}{\Weight(u)}=&\frac{\Weight(u, \Parent(u))}{\Weight(u,u) + \Weight(u, \Parent(u)) + |\Children(u)|\cdot \Weight(u,x)} 
= \frac{2^{-n}\cdot n^{-4\cdot(\Distance(u)-1)}}{O(n^{-2\cdot\Distance(u)})}\\
=&\Omega(2^{-n}\cdot n^{-2\cdot(\Distance(u)-2)})\numberthis\label{eq:weights1}
\end{align*}

\begin{align*}
\frac{\Weight(u, \Parent(u))}{\Weight(\Parent(u))}=&\frac{\Weight(u, \Parent(u))}{\Weight(\Parent(u), \Parent(u)) + \Weight(\Parent(u), \Parent(\Parent(u))) + |\Children(\Parent(u))|\cdot \Weight(u,\Parent(u))}\\
=& \frac{2^{-n}\cdot n^{-4\cdot(\Distance(u)-1)}}{\Omega(n^{-2\cdot(\Distance(u)-1)})}
=O(2^{-n}\cdot n^{-2\cdot(\Distance(u)-1)})\numberthis\label{eq:weights2}
\end{align*}

\begin{align*}
\sum_{z_i\in \Descendants(u)}\frac{\Weight(\Parent(z_i),z_i)}{\Weight(z_i)} = & |\Descendants(u)|\cdot \frac{\Weight(u,z)}{\Weight(z,z)+\Weight(u,z) + |\Children(z)|\cdot \Weight(z,z')}\\
\leq & |\Descendants(u)|\cdot \frac{2^{-n}\cdot n^{-4\cdot\Distance(u)}}{\Omega(n^{-2\cdot (\Distance(u)+1)})}
= n\cdot O(2^{-n}\cdot n^{-2\cdot (\Distance(u)-1)})\\
=& O(2^{-n}\cdot n^{-2\cdot \Distance(u)+3}) \numberthis\label{eq:weights3}
\end{align*}

\begin{comment}
\begin{align*}
\frac{\Weight(u, \Parent(u))}{\Weight(u)}=&\frac{2^{-n}\cdot n^{-4\cdot \Distance(\Parent(u))}}{n^{-2\cdot\Distance(u)}+2^{-n}\cdot n^{-4\cdot \Distance(\Parent(u))} + |\Children(u)|\cdot 2^{-n}\cdot n^{-4\cdot \Distance(u)}} 
= \frac{2^{-n}\cdot n^{-4\cdot(\Distance(u)-1)}}{O(n^{-2\cdot\Distance(u)})}\\
=&\Omega(2^{-n}\cdot n^{-2\cdot(\Distance(u)-2)})\numberthis\label{eq:weights1}
\end{align*}

\begin{align*}
\frac{\Weight(u, \Parent(u))}{\Weight(\Parent(u))}=&\frac{2^{-n}\cdot n^{-4\cdot \Distance(\Parent(u))}}{n^{-2\cdot\Distance(\Parent(u))}+2^{-n}\cdot n^{-4\cdot \Distance(\Parent(\Parent(u)))} + |\Children(\Parent(u))|\cdot 2^{-n}\cdot n^{-4\cdot \Distance(\Parent(u))}}\\
=& \frac{2^{-n}\cdot n^{-4\cdot(\Distance(u)-1)}}{\Omega(n^{-2\cdot(\Distance(u)-1)})}
=O(2^{-n}\cdot n^{-2\cdot(\Distance(u)-1)})\numberthis\label{eq:weights2}
\end{align*}

\begin{align*}
\sum_{v\in \Descendants(u)}\frac{\Weight(\Parent(v),v)}{\Weight(v)} = & |\Descendants(u)|\cdot \frac{2^{-n}\cdot n^{-4\cdot \Distance(\Parent(v))}}{n^{-2\cdot\Distance(v)}+2^{-n}\cdot n^{-4\cdot \Distance(\Parent(v))} + |\Children(v)|\cdot 2^{-n}\cdot n^{-4\cdot \Distance(v)}}\\
\leq & |\Descendants(u)|\cdot \frac{2^{-n}\cdot n^{-4\cdot\Distance(u)}}{\Omega(n^{-2\cdot (\Distance(u)+1)})}
= n\cdot O(2^{-n}\cdot n^{-2\cdot (\Distance(u)-1)})\\
=& O(2^{-n}\cdot n^{-2\cdot \Distance(u)+3}) \numberthis\label{eq:weights3}
\end{align*}
\end{comment}

Thus, using \cref{eq:weights1}, \cref{eq:weights2} and \cref{eq:weights3}, we obtain

\begin{align*}
\frac{\Probr{s_{i+1}=s_i+1}}{\Probr{s_{i+1}=s_i-1}}\geq &\frac{\frac{r}{\Fitness(\State')}\cdot \frac{\Weight(u, \Parent(u))}{\Weight(u)} }{\frac{1}{\Fitness(\State')}\cdot \left(\frac{\Weight(u, \Parent(u))}{\Weight(\Parent(u))} + \sum_{z_i\in \Descendants(u)}\frac{\Weight(\Parent(z_i),z_i)}{\Weight(z_i)}\right)}\\
=&\frac{\Omega(2^{-n}\cdot n^{-2\cdot(\Distance(u)-2)})}{O(2^{-n}\cdot n^{-2\cdot(\Distance(u)-1)})+ O(2^{-n}\cdot n^{-2\cdot \Distance(u)+3})}
=\Omega(n)
\numberthis\label{eq:ratio5}
\end{align*}

Let $\alpha(n)=1-O(n^{-1})$ and consider a one-dimensional random walk $P:s'_0, s'_1,\dots$ on states $0\leq i\leq |\Ancestors(w)| $, with transition probabilities 

\[\numberthis\label{eq:walk2}
\Probr{s'_{i+1}=\ell|s'_i}=
\left\{\begin{array}{lr}
        \alpha(n) & \text{ if } 0<s'_i<|\Sink| \text{ and } \ell=s'_i+1\\
        1-\alpha(n) & \text{ if } 0<s'_i<|\Sink| \text{ and } \ell=s'_i-1\\
        0 & \text{otherwise }
\end{array}\right.
\]

Using \cref{eq:ratio5}, we have that

\[
\frac{\Probr{s'_{i+1}=s'_i+1}}{\Probr{s'_{i+1}=s'_i-1}} = \frac{\alpha(n)}{1-\alpha(n)} = \Omega(n)  \leq \frac{\Probr{s_{i+1}=s_i+1}}{\Probr{s_{i+1}=s_i-1}}\ .
\]

Hence the probability that $s_{\infty}=|\Ancestors(w)|$ is lowerbounded by the probability that $s'_{\infty}=|\Ancestors(w)|$.
The latter event occurs with probability $1-O(n^{-1})$ (see e.g.,~\cite{Kemeny12},~\cite[Section~6.3]{Nowak06b}), as desired.
\end{proof}
\fi

\subsubsection{Analysis of Stage~3: Event $\ev_3$}
We now focus on the evolution on the hub $\Sink$, and establish several useful results.
\begin{compactenum}
\item First, we show that $G^{\Weight}\restr{\Sink}$ is isothermal (\cref{lem:sink_isothermal})

\item Second, the above result implies that the hub behaves as a well-mixed population.
Considering  advantageous mutants ($r>1$) this implies  the following (\cref{lem:sink_independent}).
\begin{compactenum}
\item Every time a mutant hits a hub of only residents, then the mutant has at least a \emph{constant} probability of fixating in the hub.
\item In contrast, every time a resident hits a hub of only mutants, then the resident has \emph{exponentially small} probability of fixating in the hub.
\end{compactenum}

\item Third, we show that an initial mutant adjacent to the hub, hits the hub a polynomial number of times (\cref{lem:sink_hit}).

\item Finally, we show that an initial mutant adjacent to the hub ensures fixating in the hub whp (\cref{lem:sink_mutant}), i.e., 
we show that given event $\ev_2$ the event $\ev_3$ happens whp.
\end{compactenum}
We start with observing that the hub is isothermal, which follows by a direct application of the definition of isothermal (sub)graphs~\cite{Lieberman05}.

\smallskip
\begin{lemma}\label{lem:sink_isothermal}
The graph $G^{\Weight}\restr{\Sink}$ is isothermal.
\end{lemma}

\begin{proof}

Consider any vertex $u\in \Sink\setminus\Frontier$. We have
\begin{align*}
\Temp\restr{X}(u)=&\sum_{v\in \Neighbors(u)\cap \Sink}\frac{\Weight\restr{\Sink}(v,u)}{\Weight\restr{\Sink}(v)} =
\frac{\Weight\restr{\Sink}(u,u)}{\Weight\restr{\Sink}(u)} + \sum_{v\in (\Neighbors(u)\setminus\{u\})\cap \Sink }\frac{\Weight\restr{\Sink}(v,u)}{\Weight\restr{\Sink}(v)} \\
=&\frac{\Weight(u,u)}{\Weight(u)} + \sum_{v\in (\Neighbors(u)\setminus\{u\})\cap \Sink }\frac{\Weight(v,u)}{\Weight(v)}\\
=&\frac{1}{\mu\cdot 2^{-n}+\nu} \cdot \left(\Weight(u,u) + \sum_{v\in (\Neighbors(u)\setminus\{u\})\cap \Sink }1\right)\\
=&\frac{1}{\mu\cdot 2^{-n}+\nu}\cdot (\mu\cdot 2^{-n}+\nu - \Degree(u) + \Degree(u))\\
=& 1
\end{align*}

since by \cref{lem:sink_weights} we have $\Weight(u)=\mu\cdot 2^{-n}+\nu$. Similarly, consider any $u\in \Frontier$. We have 
\begin{align*}
\Temp\restr{X}(u)=&\sum_{v\in \Neighbors(u)\cap \Sink}\frac{\Weight\restr{\Sink}(v,u)}{\Weight\restr{\Sink}(v)} =
\frac{\Weight\restr{\Sink}(u,u)}{\Weight\restr{\Sink}(u)} + \sum_{v\in (\Neighbors(u)\setminus\{u\})\cap \Sink }\frac{\Weight\restr{\Sink}(v,u)}{\Weight\restr{\Sink}(v)} \\
=&\frac{\Weight(u,u) + \sum_{v\in \Neighbors(u)\setminus\Sink}\Weight(u,v)}{\Weight(u)} + \sum_{v\in (\Neighbors(u)\setminus\{u\})\cap \Sink }\frac{\Weight(v,u)}{\Weight(v)}\\
=&\frac{1}{\mu\cdot 2^{-n}+\nu} \cdot \left(\Weight(u,u) +\sum_{v\in \Neighbors(u)\setminus\Sink} 2^{-n}  + \sum_{v\in (\Neighbors(u)\setminus\{u\})\cap \Sink }1\right)\\
=&\frac{1}{\mu\cdot 2^{-n}+\nu}\cdot ((\mu-|\Children(u)|)\cdot 2^{-n}+\nu - \Degree(u) + |\Children(u)|\cdot 2^{-n} + \Degree(u))\\
=& 1
\end{align*}
Thus for all $u\in \Sink$ we have $\Temp\restr{X}(u)=1$, as desired.
\end{proof}

\smallskip
\begin{lemma}\label{lem:sink_independent}
Consider that at some time $j$ the configuration of the Moran process on $G^{\Weight}$ is $\State_j$.
\begin{compactenum}
\item If $|\Sink\cap \State_j|\geq 1$, i.e., there is at least one mutant in the hub, 
then a subsequent configuration $\State_{t}$ with $\Sink\subseteq \State_t$  will be reached with probability at least $1-r^{-1} - 2^{-\Omega(n)}$
(i.e., mutants fixate in the hub with constant probability). 

\item If $|\Sink\setminus \State_j|=1$, i.e., there is exactly one resident in the hub, 
then a subsequent configuration $\State_{t}$ with $\Sink\subseteq \State_t$  will be reached with probability at least 
$1-2^{-\Omega(m)}$, where $m=n^{1-\gamma}$
(i.e., mutants fixate in the hub with probability exponentially close to~1).
\end{compactenum}

\end{lemma}
\ifshortproofs
\begin{proof}
Given a configuration $\State_i$, denote by $s_i=|\Sink \cap \State_i|$ the number of mutant individuals in the hub.
Consider any configuration $\State_i$ with $0<s_i<|\State_i|$, i.e., mutants and residents coexist in the hub.
It follows directly from the weight assignment and \cref{lem:sink_isothermal} that
the ratio of probabilities of increasing the mutants by one (i.e., $s_{i+1}=s_i+1$) over decreasing the mutants by one (i.e., $s_{i+1}=s_i-1$)
is at most
\[
\frac{1}{r} + 2^{-\Omega(n)} = \beta\ .
\]
Hence we have a one-dimensional random walk $\Walk=s_j, s_{j+1},\dots $ with backward bias $\beta=1/r + 2^{-\Omega(n)}$.
We use the standard results on one-dimensional random walks (see, e.g.,~\cite{Kemeny12},~\cite[Section~6.3]{Nowak06b}) to obtain a lowerbound on the probability that the Moran process reaches a configuration $\State_j$ with $\Sink\subseteq \State_t$. 
In particular, we have the following.
\begin{compactenum}
\item If initially $s_j=|\Sink\cap \State_j|\geq 1$ (i.e., a single mutant invades a resident hub), then the probability is
\[
\rho_1\geq \frac{1-\beta}{1-\beta^{|\Sink|}} \geq 1-r^{-1} - 2^{-\Omega(n)}\ .
\]
\item If initially $s_j=|\Sink\setminus \State_j|=1$ (i.e., a single resident invades a mutant hub), then the probability is
\[
\rho_2\geq \frac{1-\beta^{-1}}{1-\beta^{-|\Sink|}} \geq 1-2^{-\Omega(m)}\ .
\]
where $m=|\Sink|=n^{1-\gamma}$.
\end{compactenum}
The desired result follows.
\end{proof}
\fi

\iffullproofs
\begin{proof}
Given a configuration $\State_i$, denote by $s_i=|\Sink \cap \State_i|$.
Let $\State_i$ be any configuration of the Moran process with $0<s_i<|\State_i|$, 
$u$ be the random variable that indicates the vertex that is chosen for reproduction in $\State_i$,
and $\State_{i+1}$ be the random variable that indicates the configuration of the population in the next step.
By \cref{lem:sink_isothermal}, the subgraph $G^{\Weight}\restr{\Sink}$ induced by the hub $\Sink$ is isothermal, thus
\[
\frac{\Probr{s_{i+1}=s_i-1|u\in \Sink}}{\Probr{s_{i+1}=s_i+1|u\in \Sink}}=\frac{1}{r}\ . \numberthis\label{eq:ratio1}
\]
Additionally,
\begin{align*}
 \Probr{s_{i+1}=s_i-1|u\not \in \Sink} \leq & \sum_{\substack{v \in \Frontier \\ u\in\Children(v)}} \left( \frac{1}{\Fitness(\State_i)}\cdot \frac{\Weight(u,v)}{\Weight(u)} \right)
\leq n^{-1}\cdot \sum_{\substack{v \in \Frontier \\ u\in\Children(v)}} \frac{ 2^{-n}}{n^{-2}}\\
\leq&  n^{-1}\cdot n \cdot 2^{-n}\cdot n^{2} =O(n^{2}\cdot 2^{-n}) \numberthis\label{eq:ratio2}
\end{align*}

since $1/\Fitness(\State_i)\leq n^{-1}$, $\Weight(u,v)=2^{-n}$ and $\Weight(u,u)=n^{-2}$.
Moreover, as $\Sink$ is heterogeneous, it contains at least a mutant vertex $v$ and a resident vertex $w\in \Neighbors(v)$,
and $v$ reproduces with probability $r/\Fitness(\State_i)\geq n^{-1}$, and replaces the individual $v\in \Sink$
with probability at least $1/\Weight(v)$. Hence we have
\begin{align*}
\Probr{s_{i+1}=s_i+1|u\in \Sink}\cdot \Probr{u\in \Sink} \geq  \frac{1}{\Weight(u)} \cdot \frac{r}{\Fitness(\State_i)} \geq 
\frac{1}{\mu\cdot 2^{-n} + \nu} \cdot n^{-1} \geq \frac{1}{n\cdot 2^{-n} + n}\cdot n^{-1} = \Omega(n^{-2}) \numberthis\label{eq:ratio3}
\end{align*}

since by \cref{lem:sink_weights} we have $\Weight(v)=\mu\cdot 2^{-n}+\nu$.
Using \cref{eq:ratio1}, \cref{eq:ratio2} and \cref{eq:ratio3}, we have
\begin{align*}
\frac{\Probr{s_{i+1}=s_i-1}}{\Probr{s_{i+1}=s_i+1}}=&\frac{\Probr{s_{i+1}=s_i-1|u\in \Sink}\cdot \Probr{u\in \Sink} + \Probr{s_{i+1}=s_i-1|u\not \in \Sink}\cdot \Probr{u\not \in \Sink}}{\Probr{s_{i+1}=s_i+1|u\in \Sink}\cdot \Probr{u \in \Sink} + \Probr{s_{i+1}=s_i+1|u \not \in \Sink}\cdot \Probr{u \not \in \Sink}}\\
\leq& \frac{\Probr{s_{i+1}=s_i-1|u\in \Sink}\cdot \Probr{u\in \Sink} + \Probr{s_{i+1}=s_i-1|u\not \in \Sink}\cdot \Probr{u\not \in \Sink}}{\Probr{s_{i+1}=s_i+1|u\in \Sink}\cdot \Probr{u \in \Sink}}\\
\leq & \frac{\Probr{s_{i+1}=s_i-1|u\in \Sink}}{\Probr{s_{i+1}=s_i+1|u\in \Sink}} + O(n^{2})\cdot \Probr{s_{i+1}=s_i-1|u\not \in \Sink} = \frac{1}{r} + 2^{-\Omega(n)}
\numberthis\label{eq:ratio4}
\end{align*}
Hence, $s_{j}, s_{j+1},\dots$ performs a one-dimensional random walk on the states $0\leq i\leq |\Sink|$, with the ratio of transition probabilities given by \cref{eq:ratio4}.
Let $\alpha(n)=r/(r+1+2^{-\Omega(n)})$ and consider the one-dimensional random walk $\Walk:s'_{j}, s'_{j+1},\dots$ on states $0\leq i\leq |\Sink|$, with transition probabilities 

\[\numberthis\label{eq:walk1}
\Probr{s'_{i+1}=\ell|s'_i}=
\left\{\begin{array}{lr}
        \alpha(n) & \text{ if } 0<s'_i<|\Sink| \text{ and } \ell=s'_i+1\\
        1-\alpha(n) & \text{ if } 0<s'_i<|\Sink| \text{ and } \ell=s'_i-1\\
        0 & \text{otherwise }
\end{array}\right.
\]

Using \cref{eq:ratio4} we have that
\[
\frac{\Probr{s'_{i+1}=s'_i-1}}{\Probr{s'_{i+1}=s'_i+1}} = \frac{1-\alpha(n)}{\alpha(n)} = \frac{1}{r} + 2^{-\Omega(n)} \geq \frac{\Probr{s_{i+1}=s_i-1}}{\Probr{s_{i+1}=s_i+1}}\ .
\]
 
Let $\rho_1$ (resp. $\rho_2$) be the probability that the Moran process starting on configuration $\State_j$ with $|\Sink\cap \State_j|\geq 1$ (resp. $|\Sink\setminus \State_j|=1$)  will reach a configuration $\State_t$ with $\Sink\subseteq \State_t$.
We have that $\rho_1$ (resp. $\rho_2$) is lowerbounded by the probability that $\Walk$ gets absorbed in $s'_{\infty}=|\Sink|$ when it starts from $s'_{j}=1$ (resp. $s'_{j}=|\Sink|-1$).
Let 
\[
\beta=\frac{\Probr{s'_{i+1}=s'_i-1}}{\Probr{s'_{i+1}=s'_i+1}}=\frac{1}{r} + 2^{-\Omega(n)}<1\ ;
\]
and we have (see e.g.,~\cite{Kemeny12},~\cite[Section~6.3]{Nowak06b})
\[
\rho_1\geq \frac{1-\beta}{1-\beta^{|\Sink|}}\geq 1-\beta = 1-\frac{1}{r} - 2^{-\Omega(n)}\ ;
\]
and 
\[
\rho_2\geq 1-\frac{1-\beta^{-1}}{1-\beta^{-|\Sink|}} \geq 1-\frac{\beta^{-1}}{\beta^{-|\Sink|}} = 1-\beta^{|\Sink|-1}= 1-\left(\frac{1}{r} + 2^{-\Omega(n)}\right)^{n^{1-\gamma}-1}= 1-2^{-\Omega(n^{1-\gamma})}\ ;
\]
since $\beta^{-|\Sink|} > \beta^{-1}$ and thus $(\beta^{-1}-1)/(\beta^{-|\Sink|}-1)\leq \beta^{-1}/\beta^{-|\Sink|}$.
The desired result follows.
\end{proof}
\fi

\smallskip
\begin{lemma}\label{lem:sink_hit}
Consider that at some time $j$ the configuration of the Moran process on $G^{\Weight}$ is $\State_j$ such that 
$v\in \State_j$ for some $v\not\in \Sink$ that is adjacent to the hub ($\Distance(v)=1)$.
Then a mutant hits the hub at least $n^{1/3}$ times with probability $1-O(n^{-1/3})$.
\end{lemma}
\ifshortproofs
\begin{proof}
For any configuration $\State_i$ occurring after $\State_j$, let
\begin{compactenum}
\item $A$ be the event that $v$ places an offspring on $\Parent(v)$ in $\State_{i+1}$, and
\item $B$ be the event that a neighbor of $v$ places an offspring on $v$ in $\State_{i+1}$,
\end{compactenum}
and let $\rho_A$ and $\rho_B$ be the corresponding probabilities.
It follows directly from the weight assignments that
\[
\rho_A=\Omega\left( n\cdot 2^{-n} \right) \qquad \text{ and } \qquad \rho_B = O(2^{-n})\ ;
\]
i.e., the event $A$ occurs $\Omega(n)$ times more frequently than event $B$.
Using Markov's inequality, we have that with probability at least $\Omega(n^{1/3})$, the event $A$ occurs at least $n^{1/3}$ times before event $B$ occurs.
The desired result follows.
\end{proof}
\fi

\iffullproofs
\begin{proof}
For any configuration $\State_i$ occurring after $\State_j$, let
\begin{compactenum}
\item $A$ be the event that $v$ places an offspring on $\Parent(v)$ in $\State_{i+1}$, and
\item $B$ be the event that a neighbor of $v$ places an offspring on $v$ in $\State_{i+1}$,
\end{compactenum}
and let $\rho_A$ and $\rho_B$ be the corresponding probabilities.
Using \cref{eq:weights1}, we have
\[
\rho_A=\frac{r}{\Fitness(\State_i)}\cdot \frac{\Weight(v, \Parent(v))}{\Weight(v)}=\Omega\left( n\cdot 2^{-n} \right)\ ;
\numberthis\label{eq:rhoA}
\]
and using \cref{eq:weights2} and \cref{eq:weights3}
\[
\rho_B\leq \frac{r}{\Fitness(\State_i)}\cdot \left(\frac{\Weight(v, \Parent(v))}{\Weight(\Parent(u))} + \sum_{z\in \Children(v)}\frac{\Weight(v,z)}{\Weight(z)}\right)
\leq \frac{r}{n}\cdot \left(2^{-n}+O\left(n\cdot 2^{-n}\right)\right) = 2^{-\Omega(n)}\ .
\numberthis\label{eq:rhoB}
\]
since $\Parent(u)\in \Sink$ and by \cref{lem:sink_weights} we have $\Weight(\Parent(u))\geq 1$.
Let $X$ be the random variable that counts the time required until event $A$ occurs $n^{1/3}$ times.
Then, for all $\ell\in \Nats$ we have $\Prob{X\geq \ell}\leq \Prob{X'\geq \ell}$
where $X'$ is a random variable that follows the negative binomial distribution on $n^{1/3}$ failures with success rate $\rho_{X'}=1-O(n\cdot 2^{-n})\leq \rho_A$
(using \cref{eq:rhoA}).
The expected value of $X'$ is
\[
\Expect{X'} = \frac{\rho_{X'}\cdot n^{1/3}}{1-\rho_{X'}} = O\left(\frac{1-n\cdot 2^{-n}}{n^{2/3}\cdot 2^{-n}}\right)\ .
\]
Let $\alpha=2^n\cdot n^{-1/3}$, and by Markov's inequality, we have
\[
\Prob{X'\geq \alpha}\leq \frac{\Expect{X'}}{\alpha}=\frac{O\left(\frac{1-n\cdot 2^{-n}}{n^{2/3}\cdot 2^{-n}}\right)}{2^n\cdot n^{-1/3}}=O(n^{-1/3})\ .
\]
Similarly, let $Y$ be the random variable that counts the time required until event $B$ occurs.
Then, for all $\ell\in \Nats$, we have $\Prob{Y\leq \ell}\leq \Prob{Y'\leq \ell}$, where $Y'$ is a geometrically distributed variable
with rate $\rho_{Y'}=2^{-\Omega(n)}\geq \rho_B$ (using \cref{eq:rhoB}).
Then
\[
\Prob{Y'\leq \alpha} = 1-(1-\rho_{Y'})^{\alpha} = O(n^{-1/3})\ ;
\]
and thus
\[
\Prob{Y\leq  X} \leq \Prob{Y\leq \alpha} + \Prob{X\geq \alpha} \leq \Prob{Y'\leq \alpha } + \Prob{X'\geq \alpha} = O(n^{-1/3})\ .
\numberthis\label{eq:hitsink}
\]
Hence, with probability at least $1-O(n^{1/3})$, the vertex $v$ places an offspring on $\Parent(v)$ at least $n^{1/3}$ times before 
it is replaced by a neighbor.
The desired result follows.
\end{proof}
\fi

\smallskip
\begin{lemma}\label{lem:sink_mutant}
Consider that at some time $j$ the configuration of the Moran process on $G^{\Weight}$ is $\State_j$ with
$v\in \State_j$ for some $v\not\in \Sink$ that is adjacent to the hub ($\Distance(v)=1)$.
Then a subsequent configuration $\State_t$ with $\Sink \subseteq \State_t$ (mutants fixating in the hub) 
is reached with probability $1-O(n^{-1/3})$, i.e., given event $\ev_2$, the event $\ev_3$ happens whp.
\end{lemma}
\ifshortproofs
\begin{proof}
By \cref{lem:sink_hit}, we have that with probability at least $\Omega(n^{1/3})$, the vertex $v$ places an offspring on $\Parent(v)$ at least $n^{1/3}$ times before it is replaced by a neighbor.
By \cref{lem:sink_independent}, each time event a new mutant hits the hub, the probability that a configuration $\State_t$ is reached with $\Sink\subseteq \State_t$ (i.e., the hub becomes mutant) is at least
$
1-r^{-1}-2^{-\Omega(n)}
$.
Hence, the probability that a configuration $\State_t$ is reached with $\Sink\in \State_t$
is at least
\[
(1-O(n^{-1/3}))\cdot \left(1-\left(r^{-1}+z\right)^{m}\right) = 1-O(n^{-1/3})\ ;
\]
where $z=2^{-\Omega(n)}$ and $m=n^{1/3}$.
The desired result follows.
\end{proof}
\fi

\iffullproofs
\begin{proof}
By \cref{lem:sink_hit}, we have that with probability at least $\Omega(n^{1/3})$, the vertex $v$ places an offspring on $\Parent(v)$ at least $n^{1/3}$ times before it is replaced by a neighbor.
Let $t_i$ be the time that $v$ places its $i$-th offspring on $\Parent(v)$, with $1\leq i\leq n^{1/3}$.
Let $A_i$ be the event that a configuration $\State_t$ is reached, where $t\geq t_i$ and such that $\Sink\subseteq \State_t$.
By \cref{lem:sink_independent}, we have $\Probr{A_i}\geq 1-r^{-1}-2^{-\Omega(n)}$.
Moreover, with probability $1-2^{-\Omega(n)}$, at each time $t_i$ the hub is in a homogeneous state,
i.e., either $\Sink\subseteq \State_{t_i}$ or $\Sink\cap \State_{t_i}=\emptyset$.
The proof is similar to that of \cref{lem:sink_leak}, and is based on the fact that every edge which has one end on the hub and the other outside the hub has exponentially small weight (i.e., $2^{-n}$), whereas the hub $G^{\Weight}\restr{\Sink}$ resolves to a homogeneous state in polynomial time with probability exponentially close to 1.
It follows that with probability at least $p=1-2^{-\Omega(n)}$, the events $\bar{A_i}$ are pairwise independent, and thus

\[
\Probr{\overline{A}_1\cap \overline{A}_2\dots \cap \overline{A}_{n^{1/3}}} \leq p\cdot \prod_{i=1}^{n^{1/3}}\Probr{\overline{A}_i} + (1-p) \leq 
\prod_{i=1}^{n^{1/3}}(1-\Probr{A_i}) + 2^{-\Omega(n)} \leq \left(r^{-1}+2^{-\Omega(n)}\right)^{n^{1/3}} + 2^{-\Omega(n)}\ .
\numberthis\label{eq:independentAi}
\]

Finally, starting from $\State_0=\{u\}$, the probability that a configuration $\State_t$ is reached such that $\Sink\subseteq\State_t$ is lowerbounded by the probability of the events that
\begin{compactenum}
\item the ancestor $v$ of $u$ is eventually occupied by a mutant, and
\item $v$ places at least $n^{1/3}$ offsprings to $\Parent(v)\in \Sink$ before a neighbor of $v$ places an offspring on $v$, and
\item the event $\overline{A}_1\cap \overline{A}_2\dots \cap \overline{A}_{n^{1/3}}$ does not occur.
\end{compactenum}
Combining \cref{lem:ancestor_mutant}, \cref{eq:hitsink} and \cref{eq:independentAi}, we obtain that the goal configuration $\State_t$ is reached with probability at least
\[
(1-O(n^{-1}))\cdot (1-O(n^{-1/3}))\cdot \left(1-\Probr{\overline{A}_1\cap \overline{A}_2\dots \cap \overline{A}_{n^{1/3}}}\right) = 1-O(n^{-1/3})\ ;
\]
as desired.
\end{proof}
\fi

%The proof is established in two steps.
%First, \cref{lem:ancestor_mutant} guarantees that the mutant will reach the ancestor $v$ of $w$ at distance $\Distance(v)=1$ from the hub
%(i.e., the parent of $v$ belongs to the hub).
%Then, we show that whp the individual occupying $v$ will place a mutant offspring in the hub at least $n^{1/3}$ times before vertex $v$ itself is replaced.
%Since by \cref{lem:sink_independent}, each such invasion has a constant probability (at least $1-r^{-1} - 2^{-\Omega(n)} $),
%whp the hub will eventually turn mutant.

%\subsection{Fixation Probabilities on $G^{\Weight}$}\label{subsec:fixation_prob}
\subsubsection{Analysis of Stage~4: Event $\ev_4$}

%The following lemma states that the probability that a hub vertex places an offspring to a branch vertex while the hub is heterogeneous is exponentially small.
%This comes as a consequence of two facts.
%First, the hub is isothermal and evolves fast, hence each time it becomes heterogeneous, it will resolve to a homogeneous state in polynomial time, with probability exponentially close to 1.
%On the other hand, the edges between the hub and the branches have exponentially small weights, hence with probability exponentially close to 1, no such edge is used for reproduction, within any interval of polynomial length.

In this section we present the last stage to fixation.
This is established in four intermediate steps.
\begin{compactenum}
\item First, we consider the event of some vertex in the hub placing an offspring in one of the branches, while the hub is heterogeneous.
We show that this event has exponentially small probability of occurring (\cref{lem:sink_leak}).
\item We introduce the \emph{modified} Moran process which favors residents when certain events occur, more than the conventional Moran process.
This modification underapproximates the fixation probability of mutants, but simplifies the analysis.
\item We define a set of simple Markov chains $\MC_j$ and show that the fixation of mutants on the $j$-th branch $T_{m_j}^{y_j}$ is captured by the absorption probability to a specific state of $\MC_j$ (\cref{lem:coupling}).
This absorption probability is computed in \cref{lem:mc_absorb}.
\item Finally we combine the above steps in  \cref{lem:modified_process_fix} to show that if the hub is occupied by mutants (i.e., given that event $\ev_3$ holds), the mutants eventually fixate in the graph (i.e., event $\ev_4$ holds) whp.
\end{compactenum}

We start with an intermediate lemma, which states that while the hub is heterogeneous,
the probability that a node from the hub places an offspring to one of the branches is exponentially small.
\smallskip
\begin{lemma}\label{lem:sink_leak}
For any configuration $\State_j$ with $|\Sink\setminus \State_j|=1$, let $t_1\geq j$ be the first time such that $\Sink\subseteq \State_{t_1}$ (possibly $t_1=\infty$),
and $t_2\geq j$ the first time in which a vertex $u\in \Frontier$ places an offspring on some vertex $v\in\Neighbors(u)\setminus\Sink$.
We have that $\Probr{t_2<t_1}=2^{-\Omega(m)}$, where $m=n^{1-\gamma}$.
\end{lemma}

\iffullproofs
\begin{proof}
%By \cref{lem:sink_independent}, we have $t<\infty$ with probability at least $1-2^{-\Omega(n^{1-\gamma})}$.
Given a configuration $\State_i$, denote by $s_i=|\Sink \cap \State_i|$.
Recall from the proof of \cref{lem:sink_mutant} that $s_j,s_{j+1},\dots $
performs a one-dimensional random walk on the states $0\leq i\leq |\Sink|$, with the ratio of transition probabilities given by \cref{eq:ratio4}.
Observe that in each $s_i$, the random walk changes state with probability at least $n^{-2}$, which is a lowerbound on the probability that the walk progresses to $s_{i+1}=s_i + 1$ (i.e., the mutants increase by one).
Consider that the walk starts from $s_j$, and let $H_a$ be the expected absorption time, $H_f$ the expected fixation time on state $|\Sink|$, and $H_e$ the expected extinction time on state $0$ of the random walk, respectively.
The unlooped variant of the random walk $\Walk=s_i,s_{i+1},\dots$ has expected absorption time $O(n)$~\cite{Levin06},
hence the random walk $s_j,s_{j+1},\dots $ has expected absorption time
\[
H_a\leq n^{2}\cdot O(n) = O(n^3)\ ;
\]
and since by \cref{lem:sink_independent} for large enough $n$ we have  $\Probr{s_{\infty}=|\Sink|}\geq \Probr{s_{\infty}=0}$, we have
\[
H_a = \Probr{s_{\infty}=|\Sink|}\cdot H_f + \Probr{s_{\infty}=0}\cdot  H_e \implies 
H_f \leq 2\cdot H_a =O(n^3)\ .
\]
Let $t_1'$ be the random variable defined as $t_1'=t_1-j$, and we have 
\[
\Expect{t_1'|t'_1<\infty} = H_f=O(n^3)\ ;
\]
i.e., given that a configuration $\State_{t_1}$ with $\Sink\subseteq \State_{t_1}$ is reached (thus $t_1<\infty$ and $t'_1<\infty$),
the expected time we have to wait after time $j$ for this event to happen equals the expected fixation time $H_f$ of the random walk $s_j,s_{j+1},\dots$.
Let $\alpha=2^{\frac{n}{2}}$, and by Markov's inequality, we have
\[
\Probr{t_1'>\alpha|t'_1<\infty} \leq \frac{\Expect{t_1'|t'_1<\infty}}{\alpha} = n^3\cdot 2^{-\frac{n}{2}}\ .
\numberthis\label{eq:markov_bound}
\]
Consider any configuration $\State_i$. The probability $p$ that a vertex $u\in \Frontier$ places an offspring on some vertex $v\in\Neighbors(u)\setminus\Sink$ is at most
\[
p\leq \frac{r}{\Fitness(\State_i)}\cdot \sum_{u\in \Frontier}\sum_{v\in \Neighbors(u)\setminus \Sink}\frac{\Weight(u,v)}{\Weight(u)}\leq r\cdot n^{-1}\cdot n^{1-\gamma}\cdot 2^{-n} \leq r\cdot n^2 \cdot 2^{-n}\ .
\]
since $\Weight(u,v)=2^{-n}$ and by \cref{lem:sink_weights} we have $\Weight(u)>1$.
Let $t_2'=t_2-i$, and we have $\Probr{t'_2\leq\alpha}\leq \Probr{X\leq \alpha}$, where $X$ is a geometrically distributed random variable with rate $\rho= r\cdot n^2 \cdot 2^{-n}$.
Since $\Probr{t_2<t_1} = \Probr{t'_2<t'_1}$, we have
\begin{align*}
\Probr{t_2<t_1} = & \Probr{t'_2<t'_1|t'_1<\infty}\cdot \Probr{t'_1<\infty} + \Probr{t'_2<t'_1|t'_1=\infty}\cdot \Probr{t'_1=\infty}\\
\leq & \Probr{t'_2<t'_1|t'_1<\infty} + \Probr{t'_1=\infty} \\
\leq & \Probr{t'_2<t'_1|t_1<\infty} + 2^{-\Omega(n^{1-\gamma})}\\
\leq & \Probr{t'_2\leq \alpha|t'_1<\infty} + \Probr{t'_1>\alpha|t'_1<\infty} + 2^{-\Omega(n^{1-\gamma})}\\
\leq & \Probr{t'_2\leq \alpha|t'_1<\infty} + n^3\cdot 2^{-\frac{n}{2}} + 2^{-\Omega(n^{1-\gamma})}\\
\leq & \Probr{X\leq \alpha} +  2^{-\Omega(n^{1-\gamma})}\\
\leq & 1-(1-\rho)^{\alpha} + 2^{-\Omega(n^{1-\gamma})}\\
\leq & 1-(1-r\cdot n^2\cdot 2^{-n})^{2^{n/2}}+ 2^{-\Omega(n^{1-\gamma})}\\
=& 2^{-\Omega(n^{1-\gamma})}
\end{align*}

The second inequality holds since by \cref{lem:sink_independent} we have $\Probr{t'_1=\infty}=2^{-\Omega(n^{1-\gamma})}$.
The fourth inequality comes from \cref{eq:markov_bound}.
\end{proof}
\fi

%In the current section we calculate the fixation probability of mutants in $G^{\Weight}$ and conclude with our positive result.
%The key step is showing that once the hub becomes mutant, then whp all branches will become mutant before the hub becomes resident.
%Intuitively, \cref{lem:sink_independent} states that the hub has exponentially small probability of becoming resident every time it is invaded by a resident from the branches. On the other hand, because of the weight assignment, every branch has exponentially small probability of becoming mutant every time it is invaded by a mutant from the hub.
%However, the hub is polynomially larger than each of the branches, and hence the second exponential is smaller than the first.
%Hence, whp the latter event occurs before the former.
To simplify the analysis, we replace the Moran process with a \emph{modified} Moran process, which favors the residents (hence it is conservative) and allows for rigorous derivation of the fixation probability of the mutants.

\smallskip\noindent{\bf The modified Moran process.}
Consider the Moran process on $G^{\Weight}$, and assume there exists a first time $t^*<\infty$ when a configuration $\State_{t^{*}}$ is reached such that $\Sink\subseteq \State_{t^*}$.
We underapproximate the fixation probability of the Moran process starting from $\State_{t^*}$ by the fixation probability of the \emph{modified} Moran process $\ModState_{t^*},\ModState_{t^{*}+1},\dots$, which behaves as follows.
Recall that for every vertex $y_j$ with $\Distance(y_j)=1$, we denote by $T_{m_j}^{y_j}$ the subtree of $\SpanTree_n^{x}$ rooted at $y_j$, which has $m_j$ vertices.
Let $V_i$ be the set of vertices of $T_{m_i}^{y_i}$, and note that by construction $m_i\leq n^{1-c}$,
while there are at most $n$ such trees.
The \emph{modified} Moran process is identical to the Moran process, except for the following modifications.

\begin{compactenum}
\item\label{item:mod0} Initially, $\ModState_{t^*}=\Sink$.
\item\label{item:mod1} At any configuration $\ModState_i$ with $\Sink\in \ModState_i$, for all trees $T_{m_j}^{y_j}$, if a resident vertex $u\in V_j$ places an offspring on some vertex $v$ with $u\neq v$, then $\ModState_{i+1}=\ModState_i\setminus V_j$ and $|\Sink\setminus\ModState_{i+1}|=1$ i.e., all vertices of $T_{m_j}^{y_j}$ become residents and the hub is invaded by a single resident.
\item\label{item:mod2} If the modified process reaches a configuration $\ModState_i$ with $\ModState_i\cap \Sink=\emptyset$, the process instead transitions to configuration $\ModState_i=\emptyset$, i.e., if the hub becomes resident, then all mutants go extinct.
\item\label{item:mod3} At any configuration $\ModState_i$ with $\Sink\setminus\ModState_i \neq \emptyset$, if some vertex $u\in \Frontier$ places
an offspring on some vertex $v\in\Neighbors(u)\setminus\Sink$, then the process instead transitions to configuration $\ModState_i=\emptyset$,
i.e., if while the hub is heterogeneous, an offspring is placed from the hub to a vertex outside the hub, the mutants go extinct.
\end{compactenum}
Note that any time a case of \cref{item:mod0}-\cref{item:mod3} applies, the Moran and modified Moran processes transition to configurations $\State_i$ and $\ModState_i$ respectively, with $\ModState_i\subseteq \State_i$. 
\ifshortproofs
Thus,  the fixation probability of the Moran process on $G_{n}^{\Weight}$ is underapproximated by the fixation probability of the modified Moran process.
\fi
\iffullproofs
Thus,  the fixation probability of the Moran process on $G_{n}^{\Weight}$ is underapproximated by the fixation probability of the modified Moran process
(i.e., we have $\Probr{\State_{\infty}=V|t^{*}<\infty}\geq \Probr{\ModState_{\infty}=V}$).
\fi
It is easy to see that \cref{lem:sink_independent} and \cref{lem:sink_leak} directly apply to the modified Moran process.

\smallskip\noindent{\bf The Markov chain $\MC_j$.}
Recall that $T_{m_j}^{y_j}$ refers to the $j$-th branch of the weighted graph $G^{\Weight}$, rooted at the vertex $y_j$ and consisting of $m_j$ vertices.
We associate $T_{m_j}^{y_j}$ with a Markov chain $\MC_j$ of $m_j+3$ vertices, which captures the number of mutants in $T_{m_j}^{y_j}$, and whether the state of the hub.
Intuitively, a state $0\leq i \leq m_j$ of $\MC_j$ represents a configuration where the hub is homogeneous and consists only of mutants, and there are $i$ mutants in the branch $T_{m_j}^{y_j}$.
The state $\HeterState$ represents a configuration where the hub is heterogeneous, whereas the state $\DeadState$ represents a configuration where the mutants have gone extinct in the hub, and thus the modified Moran process has terminated.
We first present formally the Markov chain $\MC_j$, and later (in \cref{lem:coupling}) we couple $\MC_j$ with the modified Moran process.

Consider any tree $T_{m_j}^{y_j}$, and let $\alpha=1/(n^3+1)$.
We define  the Markov chain $\MC_j=(\MCStates_j, \trans_j)$ as follows:
\begin{compactenum}
\item The set of states is $\MCStates_j=\{\HeterState, \DeadState\}\cup \{0,1,\dots, m_j\}$
\item The transition probability matrix $\trans_j: \MCStates_j\times \MCStates_j\to [0,1]$ is defined as follows:
\begin{compactenum}
\item $\trans_j[i, i+1]=\alpha $ for $0\leq i < m_j$,
\item $\trans_j[i, 0]= 1-\alpha$ for $1<i<m_j$,
\item $\trans_j[0, \HeterState] = 1-\alpha$,
\item $\trans_j[\HeterState, 0]=1-2^{-\Omega(m)}$,  and $\trans_j[\HeterState, \DeadState]=2^{-\Omega(m)}$, where $m=n^{1-\gamma}$,
\item $\trans_j[m_j,m_j]=\trans_j[\DeadState, \DeadState]=1$,
\item $\trans_j[x,y]=0$ for all other pairs $x,y\in \MCStates_j$
\end{compactenum}
\end{compactenum}

See \cref{fig:mc_coupling} for an illustration.
\input{fig_mc_coupling}
The Markov chain $\MC_j$ has two absorbing states, $\DeadState$ and $m_j$.
We denote by $\rho_j$ the probability that a random walk on $\MC_j$ starting from state $0$ will be absorbed in state $m_j$.
The following lemma lowerbounds $\rho_j$, and comes from a straightforward analysis of $\MC_j$.
%\ifshortproofs
%We refer to the full version~\cite{FullReport} for the detailed proof.
%\fi
\smallskip
\begin{lemma}\label{lem:mc_absorb}
For all Markov chains $\MC_j$, we have  $\rho_j=1-2^{-\Omega(m)}$, where $m=n^{1-\gamma}$.
\end{lemma}

\iffullproofs
\begin{proof}
Given a state $a\in \MCStates_j$, we denote by $x_a$ the probability that a random walk starting from state $a$ will be absorbed in state $m_j$. Then $\rho_j=x_0$, and we have the following linear system

\begin{align*}
  x_{\HeterState} =& \trans[\HeterState,0] \cdot x_0 =  \left(1-2^{\Omega(n^{1-\gamma})}\right)\cdot x_0\\
%  x_{0} = & (1-\alpha)\cdot x_{\HeterState} + \alpha \cdot x_1\\
  x_{i} =& \trans[i, \HeterState] \cdot x_{\HeterState} + \trans[i,i+1] \cdot x_{i+1} = (1-\alpha)\cdot x_{\HeterState} + \alpha\cdot x_{i+1} & \text{ for } 0\leq i<m_j\\
  x_{m_j} =& 1
\end{align*}
and thus

\begin{align*}
&x_{\HeterState} = \left(1-2^{-\Omega(n^{1-\gamma})}\right)\cdot\left(x_{\HeterState}\cdot (1-\alpha)\cdot \sum_{0=1}^{m_j} a^i +a^{m_j}\right)\\
\implies & x_{\HeterState} = \left(1-2^{-\Omega(n^{1-\gamma})}\right)\cdot \left(x_{\HeterState}\cdot \left(1-a^{m_j-1}\right) + a^{m_j}\right)\\
\implies& x_{\HeterState}\left(1-\left(1-2^{-\Omega(n^{1-\gamma})}\right)\cdot \left(1-a^{m_j-1}\right)\right) = a^{m_j}
\numberthis\label{eq:xzero}
\end{align*}

Note that
\[
1-\left(1-2^{-\Omega(n^{1-\gamma})}\right)\cdot \left(1-a^{n_j-1}\right) \leq 2^{-\Omega(n^{1-\gamma})} + a^{n^j}\ ;
\]
and from \cref{eq:xzero} we obtain
\[
x_{\HeterState}\geq \frac{\alpha^{n_j}}{2^{-\Omega(n^{1-\gamma})} + \alpha^{n_j}} = 1 - \frac{2^{-\Omega(n^{1-\gamma})}}{2^{-\Omega(n^{1-\gamma})} + \alpha^{n_j}} \geq 1-2^{-\Omega(n^{1-\gamma})}\cdot \alpha^{-n_j} = 1-2^{-\Omega(n^{1-\gamma})}\cdot (n^3+1)^{n^{1-c}} = 1-2^{-\Omega(n^{1-\gamma})}\ ;
\]
since $a=1/(n^3+1)$ and by construction $n_j\leq n^{1-c}$ and  $\gamma=\eps/3<\eps/2=c$.
Finally, we have that $\rho_j=x_0\geq x_{\HeterState}=1-2^{-\Omega(n^{1-\gamma})}$, as desired.
\end{proof}
\fi

Given a configuration $\ModState_k$ of the modified Moran process, we denote by $\ov{\rho}_j(\ModState_k)$
the probability that the process reaches a configuration $\ModState_t$ with $\Sink \cup V_j\subseteq \ModState_t$.
The following lemma states that the probability $\ov{\rho}_j(\ModState_{\ell})$ is underapproximated by the probability $\rho_j$.
The proof is by a coupling argument, which ensures that
\begin{compactenum}
\item every time the run on $\MC_j$ is on a state $0\leq i\leq m_j$, there are at least $i$ mutants placed on $T_{m_j}^{y_j}$, and
\item  every time the modified Moran process transitions to a configuration where hub is heterogeneous (i.e., we reach a configuration $\State$ with $\Sink\setminus\State\neq \emptyset$),
the run on $\MC_j$ transitions to state $\HeterState$. 
\end{compactenum}

\smallskip
\begin{lemma}\label{lem:coupling}
Consider any configuration $\ModState_{\ell}$ of the modified Moran process, with $\Sink\subseteq \ModState_{\ell}$,
and any tree $T_{m_j}^{y_j}$.
We have $\ov{\rho}_j(\ModState_{\ell})\geq \rho_j$.
\end{lemma}
\begin{proof}
The proof is by coupling the modified Moran process and the Markov chain $\MC_j$.
To do so, we let the modified Moran process execute, and use certain events of that process as the source of randomness for a run in $\MC_j$. We describe the coupling process in high level.
Intuitively, every time the run on $\MC_j$ is on a state $0\leq i\leq m_j$, there are at least $i$ mutants placed on $T_{m_j}^{y_j}$.
Additionally, every time the modified Moran process transitions to a configuration where hub is heterogeneous (i.e., we reach a configuration $\State$ with $\Sink\setminus\State\neq \emptyset$),
then the run on $\MC_j$ transitions to state $\HeterState$. 
Finally, if the modified Moran process ends on a configuration $\State=\emptyset$, then the run on $\MC_j$ gets absorbed to state $\DeadState$.
The coupling works based on the following two facts.
\begin{compactenum}
\item For every state $0<i<m_j$, the ratio $\trans_j[i,i+1]/\trans_j[i,i-1]$ is upperbounded by the ratio of the probabilities of increasing the number of mutant vertices in $T_{m_j}^{y_j}$ by one, over decreasing that number by one and having the hub being invaded by a resident.
Indeed, we have 
\[
\frac{\trans_j[i,i+1]}{\trans_j[i,i-1]}=\frac{\alpha}{1-\alpha}=\frac{1}{n^3}\ ;
\]
while for every mutant vertex $x$ of $G$ with at last one resident neighbor, the probability that $x$ becomes mutant in the next step of the modified Moran process over the probability that $x$ becomes resident is at least $1/n^3$ (this ratio is at least $1/n^2$ for every resident neighbor $y$ of $x$, and there are at most $n$ such resident neighbors).
The same holds for the ratio $\trans_j[0,1]/\trans_j[0,\HeterState]$.
\item The probability of transitioning from state $\HeterState$ to state $0$ is upperbounded by the probability that once the mutant hub gets invaded by a resident the modified Moran process reaches a configuration where the hub consists of only mutants (using \cref{lem:sink_independent} and \cref{lem:sink_leak}).
\end{compactenum}
\end{proof}

The following lemma captures the probability that the modified Moran process reaches fixation whp.
That is, whp a configuration $\ModState_{i}$ is reached which contains all vertices of $G^{\Weight}$.
The proof is based on repeated applications of \cref{lem:coupling} and \cref{lem:mc_absorb}, one for each subtree $T_{m_j}^{y_j}$.

\smallskip
\begin{lemma}\label{lem:modified_process_fix}
Consider that at some time $t^*$ the configuration of the Moran process on $G^{\Weight}$ is $\State_{t^*}$ with $\Sink\subseteq \State_{t^*}$.
Then, a subsequent configuration $\State_t$ with $\State_t=V$ is reached with probability at least $1-2^{-\Omega(m)}$ where $m=n^{1-\gamma}$,
i.e., given event $\ev_3$, the event $\ev_4$ is happens whp.
\end{lemma}
\ifshortproofs
\begin{proof}
It suffices to consider the modified Moran process on $G$ starting from configuration $\ModState_{t^*}=\Sink$,
and showing that whp we eventually reach a configuration $\ModState_t=V$.
First note that if there exists a configuration $\ModState_{t'}$ with $V_i\subseteq \ModState_{t'}$ for any $V_i$,
then for all $t''\geq {t'}$ with $\ModState_{t''}\neq \emptyset$ we have $V_i\subseteq \ModState_{t''}$.
Let $t_1=t^*$.
Since $\Sink\subseteq \ModState_{t_1}$, by \cref{lem:coupling},
with probability $\ov{\rho}_1(\ModState_{t_1}) \geq \rho_1$ there exists a time $t_2\geq t_1$ such that $\Sink \cup V_1\subseteq \ModState_{t_2}$. Inductively, given the configuration $\ModState_{t_i}$, with probability $\ov{\rho}_i(\ModState_{t_i}) \geq \rho_i$ there exists a time $t_{i+1}\geq t_i$ such that $\Sink\cup V_1\cup\dots \cup V_i\subseteq \ModState_{t_{i+1}}$.
Since $V=\Sink\cup (\bigcup_{i=1}^k V_i)$, we obtain that the probability that the mutants get fixed is at least
\[
\prod_{i=1}^{n} \rho_i \geq \left(1-2^{-\Omega(m)}\right)^n=1-2^{-\Omega(m)}\ ;
\]
as by \cref{lem:mc_absorb} we have that $\rho_i=1-2^{-\Omega(m)}$ for all $i$.
The desired result follows.
\end{proof}
\fi
\iffullproofs
\begin{proof}
It suffices to consider the modified Moran process on $G$ starting from configuration $\ModState_{t^*}=\Sink$,
and showing that whp we eventually reach a configuration $\ModState_t=V$.
First note that if there exists a configuration $\ModState_{t'}$ with $V_i\subseteq \ModState_{t'}$ for any $V_i$,
then for all $t''\geq {t'}$ with $\ModState_{t''}\neq \emptyset$ we have $V_i\subseteq \ModState_{t''}$.
Let $t_1=t^*$.
Since $\Sink\subseteq \ModState_{t_1}$, by \cref{lem:coupling},
with probability $\ov{\rho}_1(\ModState_{t_1}) \geq \rho_1$ there exists a time $t_2\geq t_1$ such that $\Sink \cup V_1\subseteq \ModState_{t_2}$. Inductively, given the configuration $\ModState_{t_i}$, with probability $\ov{\rho}_i(\ModState_{t_i}) \geq \rho_i$ there exists a time $t_{i+1}\geq t_i$ such that $\Sink\cup V_1\cup\dots \cup V_i\subseteq \ModState_{t_{i+1}}$.
Since $V=\Sink\cup (\bigcup_{i=1}^k V_i)$, we obtain
\[
\Probr{\ModState_{\infty}=V} \geq \prod_{i=1}^{n} \rho_i = \prod_{i=1}^{n} \left(1-2^{-\Omega(n^{1-\gamma})}\right) \geq \left(1-2^{-\Omega(n^{1-\gamma})}\right)^n=1-2^{-\Omega(m)}\ ;
\]
as by \cref{lem:mc_absorb} we have that $\rho_i=1-2^{-\Omega(m)}$ for all $i$.
The desired result follows.
\end{proof}
\fi

\subsubsection{Main Positive Result}

We are now ready to prove the main theorem of this section.
First, combining \cref{lem:initialization}, \cref{lem:ancestor_mutant}, \cref{lem:sink_mutant} and \cref{lem:modified_process_fix}, we obtain that if $r>1$, then the mutants fixate $G_n$ whp.

\smallskip
\begin{lemma}\label{lem:amplifiers}
For any fixed $\eps >0$, for any graph $G_n$ of $n$ vertices and diameter $\Diameter(G_n)\leq n^{1-\eps}$,
there exists a weight function $\Weight$ such that for all $r>1$, we have $\FixProb(G^{\Weight}_n,r,\Uniform)=1-O(n^{-\eps/3})$ and $\FixProb(G^{\Weight}_n,r,\Temp)=1-O(n^{-\eps/3})$.
\end{lemma}

It now remains to show that if $r<1$, then the mutants go extinct whp.
This is a direct consequence of the following lemma, which states that for any $r\geq 1$, the fixation probability of a mutant with relative fitness $1/r$ is upperbounded by one minus the fixation probability of a mutant with relative fitness $r$, in the same population.

\smallskip
\begin{lemma}\label{lem:uniform_suppress_residents}
For any graph $G_n$ and any weight function $\Weight$, for all $r\geq 1$, we have that  $\FixProb(G^{\Weight}_n,1/r,\Uniform) \leq  1-\FixProb(G^{\Weight}_n,r,\Uniform)$.
\end{lemma}
\begin{proof}
Let $\sigma$ be any irreflexive permutation of $V$ (i.e., $\sigma(u)\neq u$ for all $u\in V$), and observe that for every vertex $u$, the probability that a mutant of fitness $1/r$ arising at $u$ fixates in $G_n$ is upperbounded by one minus the probability that a mutant of fitness $r$ arising in $\sigma(u)$ fixates in $G_n$. We have
\begin{align*}
\FixProb(G^{\Weight}_n,1/r,\Uniform) =& \frac{1}{n}\sum_{u}\FixProb(G^{\Weight}_n,1/r,u)\\
\leq& \frac{1}{n}\cdot \sum_{u} (1- \FixProb(G^{\Weight}_n,r,\sigma(u)))\\
=&1- \frac{1}{n}\cdot \sum_{\sigma(u)}\FixProb(G^{\Weight}_n,r,u)\\
=& 1-\FixProb(G^{\Weight}_n,r,\Uniform)
\end{align*}
\end{proof}

A direct consequence of the above lemma is that under uniform initialization, for any graph family where the fixation probability of advantageous mutants ($r>1$) approaches $1$, the fixation probability of disadvantageous mutants ($r<1$) approaches zero.
Since under our weight function $\Weight$ temperature initialization coincides with uniform initialization whp,
\cref{lem:amplifiers} and \cref{lem:uniform_suppress_residents} lead to the following corollary,
which is our positive result.

\smallskip
\begin{theorem}\label{them:amplifiers}
Let $\eps>0$ and $n_0>0$ be any two fixed constants, and
consider any sequence of unweighted, undirected graphs $(G_n)_{n>0}$ such that $\Diameter(G_n)\leq n^{1-\eps}$ for all $n>n_0$.
There exists a sequence of weight functions $(\Weight_n)_{n>0}$ such that the graph family
$\Graphseq=(G_n^{\Weight_n})$ is a (i)~strong uniform, (ii)~strong temperature, and (iii)~strong convex amplifier.
\end{theorem}

%% file: fig_mc_coupling.tex
\begin{figure}[!h]
\centering
\begin{tikzpicture}[thick,
pre/.style={<-,shorten >= 1pt, shorten <=1pt, thick,  bend angle = 20},
post/.style={->,shorten >= 1pt, shorten <=1pt,  thick,  bend angle = 20},
every loop/.style={<-,shorten >= 1pt, shorten <=1pt, thick, auto, in=110,out=70, looseness=8, min distance=14mm},
state/.style={circle,draw=black!80, line width=1.5pt, inner sep=2pt,minimum size=30pt},
target state/.style={circle, draw=black!80, line width = 1.2pt, minimum size = 25pt},
virt/.style={circle,draw=white!50,fill=white!20,thick, inner sep=2pt,minimum size=30pt},
node distance=2.5cm]

\newcommand{\xdisp}{2.7}
\newcommand{\ydisp}{1.5}
\newcommand{\bend}{20}
\newcommand{\bendtwo}{45}
\node	[state]	(d)	at	(0*\xdisp,0)	{$\DeadState$};
\node [state, right of=d]	(h)	{$\HeterState$};
\node [state, right of=h]	(0)	{$0$};
\node [state, right of=0]	(1)	{$1$};
\node [state, right of=1]	(2)	{$2$};
\node[virt, right of=2]	(dots)	{$\dots$};
\node [state, right of=dots]	(nj)	{$n_j$};

\draw[post, bend right=\bend] (h) to node[auto, above]	{$2^{-\Omega(m)}$} 	(d);
\draw[post, bend right=\bend] (h) to node[auto, below]	{$1-2^{-\Omega(m)}$} (0);
\draw[post, bend right=\bend] (0) to node[auto, above]	{$1-\alpha$} (h);
\draw[post, bend right=\bend] (0) to node[auto, below]	{$\alpha$} (1);
\draw[post, bend right=\bend] (1) to node[auto, above]	{$1-\alpha$} (h);
\draw[post, bend right=\bend] (1) to node[auto, below]	{$\alpha$} (2);
\draw[post, bend right=\bendtwo] (2) to node[auto, above]	{$1-\alpha$} (h);
\draw[post, loop above] (d) to node[auto, above]	{$1$} (d);
\draw[post, loop above] (nj) to node[auto, above]	{$1$} (nj);
%\draw[post, bend right=\bend] (2) to node[auto, below]	{$\alpha$} (dots);

\end{tikzpicture}
\caption{The Markov chain $\MC_j$ given a tree $T_{n_j}^{x_j}$.}\label{fig:mc_coupling}
\end{figure}
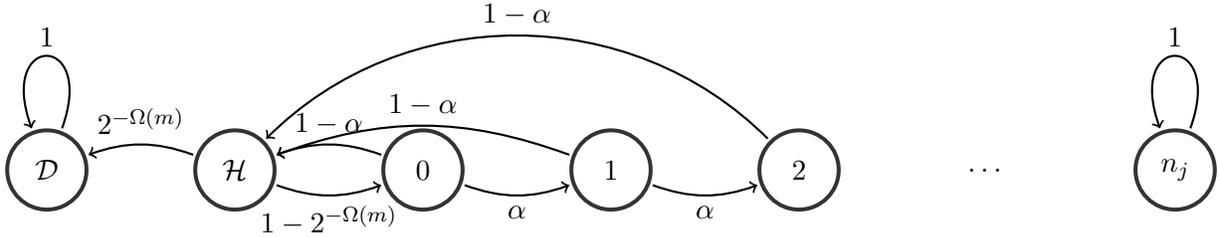